\documentclass[10pt]{article}
\usepackage{a4wide}
\usepackage{multicol}
\usepackage{verbatim}
\usepackage{subfigure}
\usepackage{bbm,amssymb}
\usepackage{amsfonts}
\usepackage{amsmath}
\usepackage{pictexwd}
\usepackage{amsthm}
\usepackage{graphicx}
\usepackage{natbib}
\usepackage{rotating}
\usepackage{color}

\newcommand{\dickm}[1]{\text{\boldmath ${#1}$}}

\theoremstyle{plain}
\newtheorem{thm}{Theorem}[section]

\newtheorem{satz}[thm]{Proposition}  

\newtheorem{defi}[thm]{Definition}
\newtheorem{bem}[thm]{Remark}
\newtheorem{rems}[thm]{Remarks}

\begin{document}

\title{\LARGE The Effect of Recurrent Mutation on the Linkage Disequilibrium under a Selective Sweep}

\thispagestyle{empty}

\author{C. Borck\thanks{Albert-Ludwigs-University, Freiburg, Eckerstra{\ss}e 1, 79104 Freiburg, Deutschland, e-mail: Cornelia.Borck@stochastik.uni-freiburg.de}}

\maketitle

\begin{abstract}
A selective sweep describes the reduction of diversity due to strong positive selection.
If the mutation rate to a selectively beneficial allele is sufficiently high, \citet{PenningsHermisson2006a} have shown, that it becomes likely, that a
 selective sweep is caused by several individuals. Such an event is called a soft sweep and the complementary event of a single origin of the beneficial
allele, the classical case, a hard sweep. 
We give analytical expressions for the linkage disequilibrium (LD) between two neutral loci linked to the selected locus, depending on the recurrent mutation 
to the beneficial allele, measured by $D$ and $\widehat{\sigma_D^2}$,
a quantity introduced by \citet{Ohta1969}, and conclude that the LD-pattern of a soft sweep differs substantially from that of a hard sweep due to haplotype structure.
We compare our results with simulations.
\end{abstract}

\section{Introduction}
It is a long-standing question of evolutionary biology to decide about the relative importance of evolutionary factors such as selection versus genetic drift to shape
 patterns of diversity. 
Today this topic is studied based on DNA variation data taken from a sample of a population.
Important work on the effects of positive selection on patterns in DNA data was made by \citet{MaynardSmithHaigh1974}. They showed, that neutral variation linked to a beneficial
 allele also increases in frequency. This is called the hitchhiking effect and the resulting reduction of neutral variation is termed a selective sweep.  
When a beneficial allele fixes in a population, this allele can have a
 single or several origins, i.e.\ it can be brought to the population
 by a single or by several mutants. If several individuals, called founders, account
 for the fixation of a beneficial allele, we will talk as Pennings and
 Hermisson about a soft sweep and else about a hard sweep.

There are various reasons for a soft sweep. Adaptation can occur from recurrent migration, mutation or act on standing genetic variation.
 We treat here the case of recurrent 
mutation, which also applies to migration in a special case. Realistic models for recurrent migration in general can lead to more complex scenarios due to population structure.

It has been shown by Hermisson and Pennings in \citep{HermissonPennings2005,PenningsHermisson2006a}, that soft sweep events become frequent, if the scaled
 mutation rate $\theta_s= 4Nu_s$ (where $N$ is the diploid population size and $u_s$ the mutation probability to the beneficial allele per individual per generation) is sufficiently high. 
While hard sweeps dominate for $\theta_s<0.01$, both hard sweeps and soft sweeps occur in the range $0.01<\theta_s < 1$. For $\theta_s>1$ almost
 all adaptive substitutions will result in soft sweeps. Soft sweeps become likely for populations with large population sizes $N$ or for alleles with high recurrent mutation rates $u_s$.
 For example, most pathogens have extremely high population sizes. Therefore their genomes are good candidates for the detection of
soft selective sweeps, see e.g. \citep{NairEtAl2007} for research on soft sweeps in malaria parasites.
 \citet{Karasov2010} concluded lately that in Drosophila melanogaster there should exist a huge amount of soft sweeps due to tremendous short-term effective population
sizes relevant for adaptation. \citet{Schlenke2005} located some of these regions. Recent research by \citet{Scheinfeldt2009} shows, that the DNA pattern around the human
 gene ALMS1, causing the Alstroem Syndrome which presents with early childhood obesity and insulin resistance leading to Type 2 diabetes,
 may also be the result of a soft sweep.
Further \citet{TishkoffEtAl2007} found out, that different SNPs in the human genome all lying in the same short genome region of 110 bp are responsible
 for the human lactase persistence in the African and European human populations. In their studies
 of LD (measured by the D' value and the LOD score) the pattern
of a soft sweep can be recognized. Ongoing research argues for the importance of soft sweeps and polygenic adaptation, see for a review about this discussion e.g. \citep{Pritchard2010}. 

In order to detect soft sweeps it is important to understand the footprints they leave in DNA data.
For this purpose it is necessary to make
 statistical predictions available, which allow us to find targets
 of recent positive selection. 
\citet{PenningsHermisson2006a} showed that tests based on haplotype structure have high power to detect
soft sweeps. If a soft sweep occurred, the population can be divided into several haplotype groups, one for each founder. Without mutation and recombination during the sweep the genomes of the groups differ at 
the same loci as the founders differed in the beginning of the sweep. Especially, in the case of two founders each allelic 
variant of a SNP locus is always linked to a single haplotype group. So high linkage disequilibrium of 
two neutral loci in a neighborhood of the selected locus should be found.
This gives rise to the conjecture, that linkage disequilibrium is a useful quantity to detect soft sweeps. 

\smallskip
LD has been computed under neutrality by
\citet{Ohta1969}.
\citet{Stephan2006}, \citet{McVean2007} and \citet{PfaffelhuberLehnertStephan2008} gave analytical expressions for measures 
of LD after a hard selective sweep.  
\citet{Kim2004} developed a composite-likelihood method for detecting hard sweeps incorporating information
from measures of linkage disequilibrium based on simulation studies.
The aim of this article is to give analytical expressions for linkage disequilibrium under a 
selective sweep with recurrent mutation to the beneficial allele (see Theorem \ref{smallthm} and Theorem \ref{bigthm}). 
To determine the linkage disequilibrium we use an extended star-like
 approximation for the genealogy of the selected site, see Section \ref{GeneaDefi}. A similar approach was applied by \citet{PfaffelhuberLehnertStephan2008} to obtain 
the linkage disequilibrium in the case of a hard sweep. We will see, that soft sweeps 
produce a different signal than hard sweeps. In Section \ref{sim} we compare our computations with simulations.

\begin{figure}
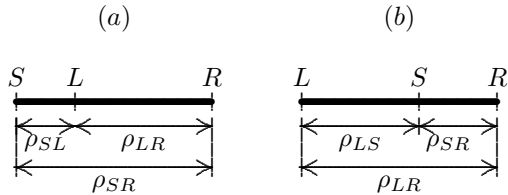

  \text{
    \parbox{1.5in}{\beginpicture
      \setcoordinatesystem units <1.3cm,1.1cm>
      \setplotarea x from 0 to 1.7, y from 0 to 2
      \plot 0 1.5 2 1.5 /
      \plot 0 1.6 0 1.4 /
      \plot 0.6 1.6 0.6 1.4 /
      \plot 2 1.6 2 1.4 /
      \put{$S$} [cC] at 0 1.8
      \put{$L$} [cC] at 0.6 1.8
      \put{$R$} [cC] at 2 1.8
      \plot 0 1.3 0 1.1 /
      \plot 0 0.6 0 1.3 /
      \plot 2 0.6 2 1.3 /
      \multiput {\tiny $\bullet$} at 0 1.5 *200 .01 .0 /
      \arrow <0.2cm> [0.375,1] from 0 1.2 to .6 1.2
      \arrow <0.2cm> [0.375,1] from 0.6 1.2 to 0 1.2
      \arrow <0.2cm> [0.375,1] from 2 1.2 to .6 1.2
      \arrow <0.2cm> [0.375,1] from 0.6 1.2 to 2 1.2
      \arrow <0.2cm> [0.375,1] from 2 0.7 to 0 0.7
      \arrow <0.2cm> [0.375,1] from 0 0.7 to 2 0.7
      \put{$\rho_{SL}$} [cC] at .3 1
      \put{$\rho_{LR}$} [cC] at 1.3 1
      \put{$\rho_{SR}$} [cC] at 1 .5
      \put{$(a)$} [cC] at 1 2.5
      \endpicture}} 
  \qquad
  \text{
    \parbox{1.5in}{\beginpicture
      \setcoordinatesystem units <1.3cm,1.1cm>
      \setplotarea x from 0 to 2, y from 0 to 2
      \plot 0 1.5 2 1.5 /
      \plot 0 1.6 0 1.4 /
      \plot 1.2 1.6 1.2 1.4 /
      \plot 2 1.6 2 1.4 /
      \put{$L$} [cC] at 0 1.8
      \put{$S$} [cC] at 1.2 1.8
      \put{$R$} [cC] at 2 1.8
      \plot 0 1.3 0 0.6 /
      \plot 1.2 1.3 1.2 1.1 /
      \plot 2 1.3 2 0.6 /
      \multiput {\tiny $\bullet$} at 0 1.5 *200 .01 .0 /
      \arrow <0.2cm> [0.375,1] from 0 1.2 to 1.2 1.2
      \arrow <0.2cm> [0.375,1] from 1.2 1.2 to 0 1.2
      \arrow <0.2cm> [0.375,1] from 2 1.2 to 1.2 1.2
      \arrow <0.2cm> [0.375,1] from 1.2 1.2 to 2 1.2
      \arrow <0.2cm> [0.375,1] from 2 0.7 to 0 0.7
      \arrow <0.2cm> [0.375,1] from 0 0.7 to 2 0.7
      \put{$\rho_{LS}$} [cC] at .6 1
      \put{$\rho_{SR}$} [cC] at 1.5 1
      \put{$\rho_{LR}$} [cC] at 1 .5
      \put{$(b)$} [cC] at 1 2.5
      \endpicture}} 
  \caption{\label{fig:geom}The two possible geometries of the selected locus
    ($S$) and the two neutral loci ($L$ and $R$). The scaled
    recombination rates between the two loci are given by $\rho_{SL}$,
    $\rho_{LR}$, $\rho_{LS}$ and $\rho_{SR}$.}
\end{figure}

\section{Model and measures of linkage disequilibrium}
\subsection{The frequency of the selected locus during a sweep}\label{freq}
We consider a DNA region of a population of $N$ diploid individuals and concentrate on the neighborhood of a bi-allelic selected locus S
 with a wild-type allele $b$ and beneficial allele $B$. The new beneficial allele $B$ with fitness advantage $s$ enters the population recurrently by mutation
 and is assumed to fix eventually.
 The population reproduces at the beneficial locus according to the Moran model in continuous time with selection and recurrent
mutation to the beneficial allele, i.e. denoting by $(X^{N}(t))_{t\geq0}$ the frequency of the individuals
 carrying the beneficial allele in a population of size $N$, 
$(X^{N}(t))_{t\geq0}$ is a jump Markov process with transition rates from

\begin{align}
& i/2N \textrm{ to } (i + 1)/2N \textrm{ at rate } u_s (2N-i) + (2N-i) \frac{i}{2N}(1+s)  \\
& i/2N \textrm{ to } (i-1)/2N \textrm{ at rate }  i \frac{(2N-i)}{2N},  
\end{align}
with $u_s, s \geq  0$.
The rate $u_s (2N-i)$ is the mutation rate, the rates $(2N-i) \frac{i}{2N}(1+s)$,  $i \frac{(2N-i)}{2N}$, respectively, are resampling rates which change
the frequency of the beneficial allele by plus $\frac{1}{2N}$, minus $\frac{1}{2N}$, respectively. Of course, resampling events inside
the beneficial and wild-type locus, which do not change the frequency of the beneficial allele, are also possible.

The frequency of the beneficial allele can be approximated for large $N$ by a differential equation:
\begin{satz}\label{kurzthm}
Denote by $X^N (t)$ the frequency of the beneficial allele in the Moran model with constant diploid population size $N$ at time $t$. 
Let $X^{N}(0)= \frac{\lceil \epsilon2N \rceil}{2N}$. Then the frequency of the beneficial allele
$X^N (t)$ converges for $N\rightarrow \infty$ to the solution of the differential equation:  
\begin{equation}\label{diffeq}
\dot{X}(t) = u_s (1-X(t))+ s X(t)(1-X(t)) 
\end{equation}
with initial condition
$$X(0)= \epsilon,$$
in the sense, that for all $\delta >0$ and all $t$
$$\lim \limits_{N \rightarrow \infty} P(\sup \limits_{s\leq t} |X^N(s) -X(s)| >\delta) =0.$$

\end{satz}
The proof is an easy application of Theorem 3.1 in \citep{Kurtz1971}. 

Equation (\ref{diffeq}) has the solution
\begin{equation}
 X(t)= \frac{(\epsilon s + u_s)s e^{s t} - (s - \epsilon s)
    u_s e^{
     -u_s t}}{ s ((\epsilon s + u_s )
      e^{s t} + (s - \epsilon s) e^{-u_s t})}.
\end{equation}
In this approximation we say that the allele fixes in the population,
 if $X_{T}=1-\epsilon$. For the above equation,
 this happens at time $T= \frac{1}{s+u_s} \log\Big(\frac{(1-\epsilon)(u_s+s(1-\epsilon))}{\epsilon(\epsilon s+u_s)}\Big)$. 
With $\epsilon= 1/\alpha$, we obtain $T=\frac{1}{s+u_s} \log\Big(\frac{(\alpha-1)(\alpha-1+(\theta_s/2))}{1+(\theta_s/2)}\Big)$, denoting
by $\alpha:=2Ns$ and $\theta_s:=4Nu_s$. For small $u_s \ll s$ and large $N$, such that $\alpha\gg 1$, the fixation time $T$ is approximately
$2(\log{\alpha})/s$.
The fixation time will be relevant below.

\subsection{Measures of linkage disequilibrium} 
Our aim is to provide analytical results for the linkage disequilibrium of two neutral loci in a neighborhood of the selected locus.
 Different quantities have been proposed to measure the association of two loci. We will compute two of them here.
 Consider two neutral loci $L$ and $R$, linked to the selective locus $S$.
 The neutral loci can either lie both on the same side of the selected locus or the selected locus
 lies between the neutral loci, see Figure \ref{fig:geom}. (If both neutral loci lie on the left side of the selected locus, we name the leftmost locus $R$-locus
and the locus in the middle $L$-locus, i.e. we have the ordering $R$ $L$ $S$.)
We consider only loci with exactly two allelic
variants. Denote them by $L/\ell$ and $R/r$ and their allelic frequencies by $q_{\ell}, q_L$, etc.

\begin{defi}[$D_{\ell,r}$ and $\widehat{D_{\ell,r}}$ ]

The simplest approach to measure linkage disequilibrium between the allele $\ell$ of the L-locus and the allele $r$ of the R-locus is to compute the quantity
\begin{equation}
 D_{\ell,r} :=  q_{\ell r} - q_{\ell} q_r.
\end{equation}
If $D_{\ell, r}$ is zero the alleles $\ell$ and $r$ are said to be in linkage equilibrium, else in linkage disequilibrium.

In practice, the population frequencies $q_{\ell},  q_r$, etc.~are often not available, but only the allelic frequencies
 in a sample, $\widehat{q}_{\ell}$,$\widehat{q}_r$, etc.
In samples LD can be measured by 
$$\widehat{D_{\ell,r}}:= \widehat{q}_{\ell r} - \widehat{q}_{\ell} \widehat{q}_r.$$ 
\end{defi}
\begin{bem} \label{D}
\emph{An easy calculation shows, that
\begin{equation*} 
D_{\ell, r}= D_{L,R}= - D_{\ell,R}=- D_{L,r}
\end{equation*}
and analogous equalities hold for $\widehat{D}$.}
\end{bem}

Averaging $D_{\ell,r}$ over all allelic variants gives zero due to Remark \ref{D}.
Hence it makes sense to consider
 $D_{\ell,r}^2$.
Since $D_{\ell,r}^2= D_{\ell,R}^2= D_{L,R}^2= D_{L,r}^2$ the quantity $D_{\ell,r}^2$ actually does not depend on the allelic
 variant $\ell$,$r$. Therefore we write $D^2$ instead. 
However $D^2$ depends strongly on the size of the allelic variants: If $q_{\ell}$ and $q_r$ are small, $D^2$ is also small, whenever the
allelic variants may be not in association at all. 
Therefore the so called standard linkage disequilibrium, introduced by \citet{Ohta1969}, is often considered.
\begin{defi}[$\sigma^2_D$ and $\widehat{\sigma^2_D}$]
The standard linkage disequilibrium $\sigma^2_D$, $\widehat{\sigma^2_D}$ in the sample, respectively, is given by
$$ \sigma^2_D = \frac{\mathbbm{E}[D_{l,r}^2]}{\mathbbm{E}[q_{\ell} (1-q_{\ell}) q_r(1-q_r)]},$$ and
$$\widehat{\sigma^2_D} = \frac{\mathbbm{E}[\widehat{D_{l,r}}^2]}{\mathbbm{E}[ \widehat{q}_{\ell} (1-\widehat{q}_{\ell}) \widehat{q}_r (1-\widehat{q}_r)]},$$
respectively.
\end{defi}

\begin{rems}
\emph{
\begin{itemize}
\item Note, that $\frac{\mathbbm{E}[\widehat{D_{l,r}}^2]}{\mathbbm{E}[ \widehat{q}_{\ell} (1-\widehat{q}_{\ell}) \widehat{q}_r(1-\widehat{q}_r)]}$ does not 
depend on the particular allelic variant $\ell, r$, too. Therefore it makes sense to write $\widehat{\sigma^2_D}$, instead of $\widehat{\sigma^2_{D_{\ell,r}}}$
\item We compute linkage disequilibrium during the sweep. So, if time is important, we will write $D_{\ell, r}(t):=  q_{\ell r}(t) - q_{\ell}(t) q_r(t)$, etc.
\item \citet{PfaffelhuberLehnertStephan2008} have computed $\mathbbm{E}[D_{\ell, r}(0)| D_{\ell,r}(T)]$ and $\widehat{\sigma^2_D}$ for a hard sweep, i.e. the case $u_s=0$.
See Figure \ref{fig:Theo} for a plot of $\widehat{\sigma^2_D}$ under neutrality and for $\theta_s=0$ and $\theta_s= 0.1$.
\item Naturally one would consider the quantity $$r^2:=  \mathbbm{E}\Big[\frac{\widehat{D}^2}{\widehat{q}_L(1-\widehat{q}_L) \widehat{q}_R(1-\widehat{q}_R)}\Big].$$ But this
quantity is less attractive for analytical studies, because it is mathematically difficult to handle. However, see the recent paper of \citet{SongSong2007} for an analytical computation of $r^2$ under
neutrality.  
\end{itemize}}
\end{rems}

\subsection{Genealogies: Motivation}
We want to compute $\mathbbm{E}[D_{\ell,r}(0)|D_{\ell,r}(T)]$ and $\widehat{\sigma^2_D}$ at
 the end of the sweep assuming small sample sizes $n\ll N$ for the computation of $\widehat{\sigma^2_D}$.
We will use the 1-1-correspondence between the probability to draw two pairs of heterozygous neutral loci and $\widehat{\sigma^2_D}$, (see 
step 3 of the proof of Theorem \ref{bigthm}). The probability to draw a heterozygous pair at the end of the sweep differs from the probability to draw a heterozygous pair
at the beginning of the sweep due to the change of the genealogy during the sweep. We shall start with heterozygous pairs of a sample 
taken from the population at the end of the sweep and follow the lines of the pairs till the beginning of 
the sweep. In our notation time is running backwards starting from time $T$ of fixation, i.e. if $t_2> t_1$ the time $t_2$ lies further back in the past then the time $t_1$.  For example 
$\mathbbm{E}[D_{\ell,r}(0)|D_{\ell,r}(T)]$ is the expected value of $D_{\ell,r}$ at the end of the sweep given $D_{\ell,r}$ at the beginning of the sweep.

To define the genealogies of two neutral loci in the neighborhood of a selected locus in a Moran model we would have to extend the Moran model as introduced in
Section \ref{freq} to a full three-locus model. However, multi-locus genealogies under such a Moran model are very complex. Under  
certain conditions star-like genealogies approximate the genealogies
 of the Moran model quite well and allow a computation of the above probabilities due to independent genealogical lines.    
In the following we introduce such star-like genealogies and justify why it is reasonable to use them in our setting.  

We suppose, that neutral mutations occur according to a Poisson Process with rates of order $\mathcal{O}(1/N)$. 
Since the sweep takes only of order $\log(2Ns)/s$ time units, we can ignore neutral mutations during the sweep.
Moreover, back-mutations are rapidly sorted out as they have no fitness advantage. Hence we will ignore back-mutations, too.

\begin{figure}
   \subfigure
             {
              \includegraphics[scale=0.55]{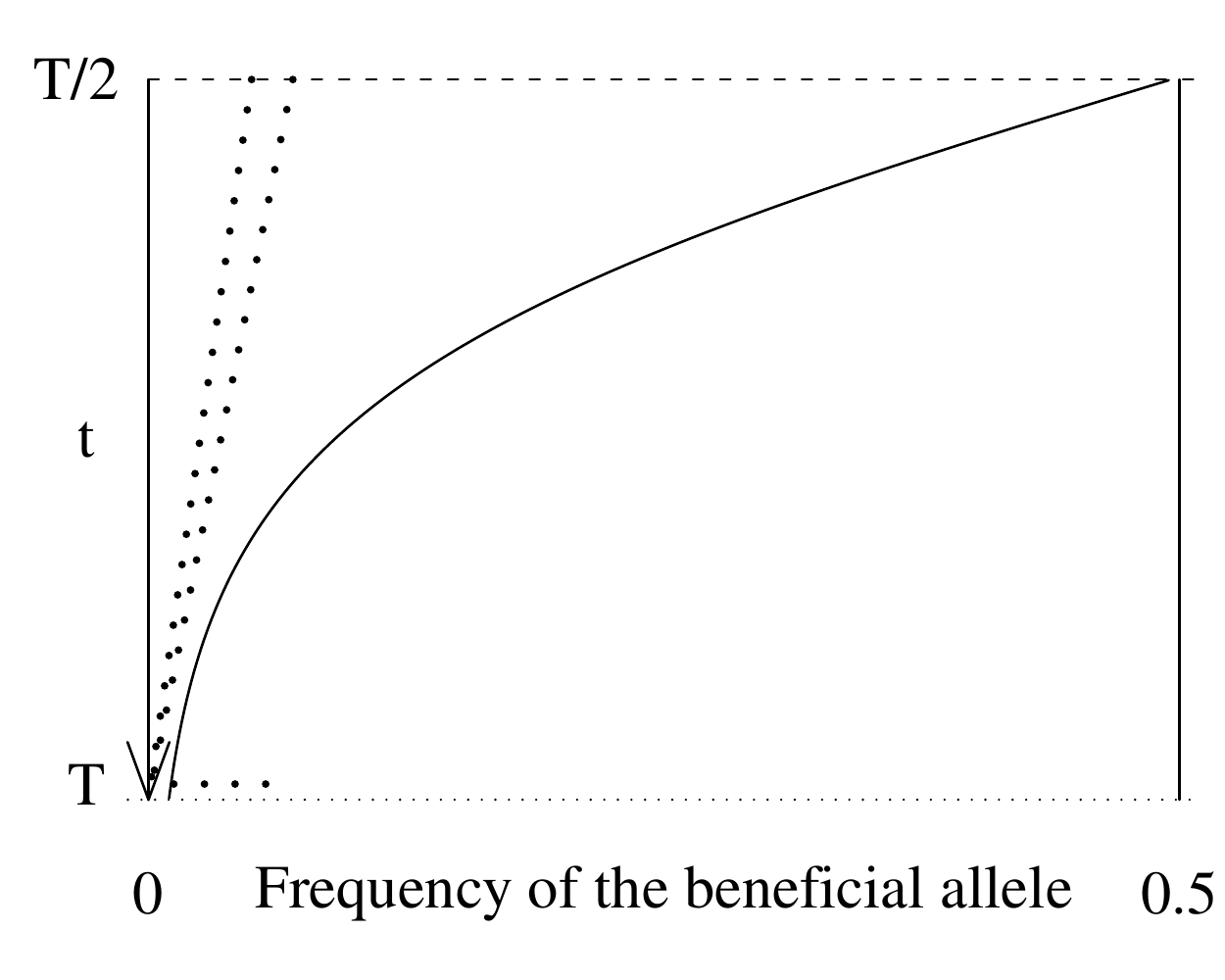}
             }
   \subfigure
             {
            \includegraphics[scale=0.55]{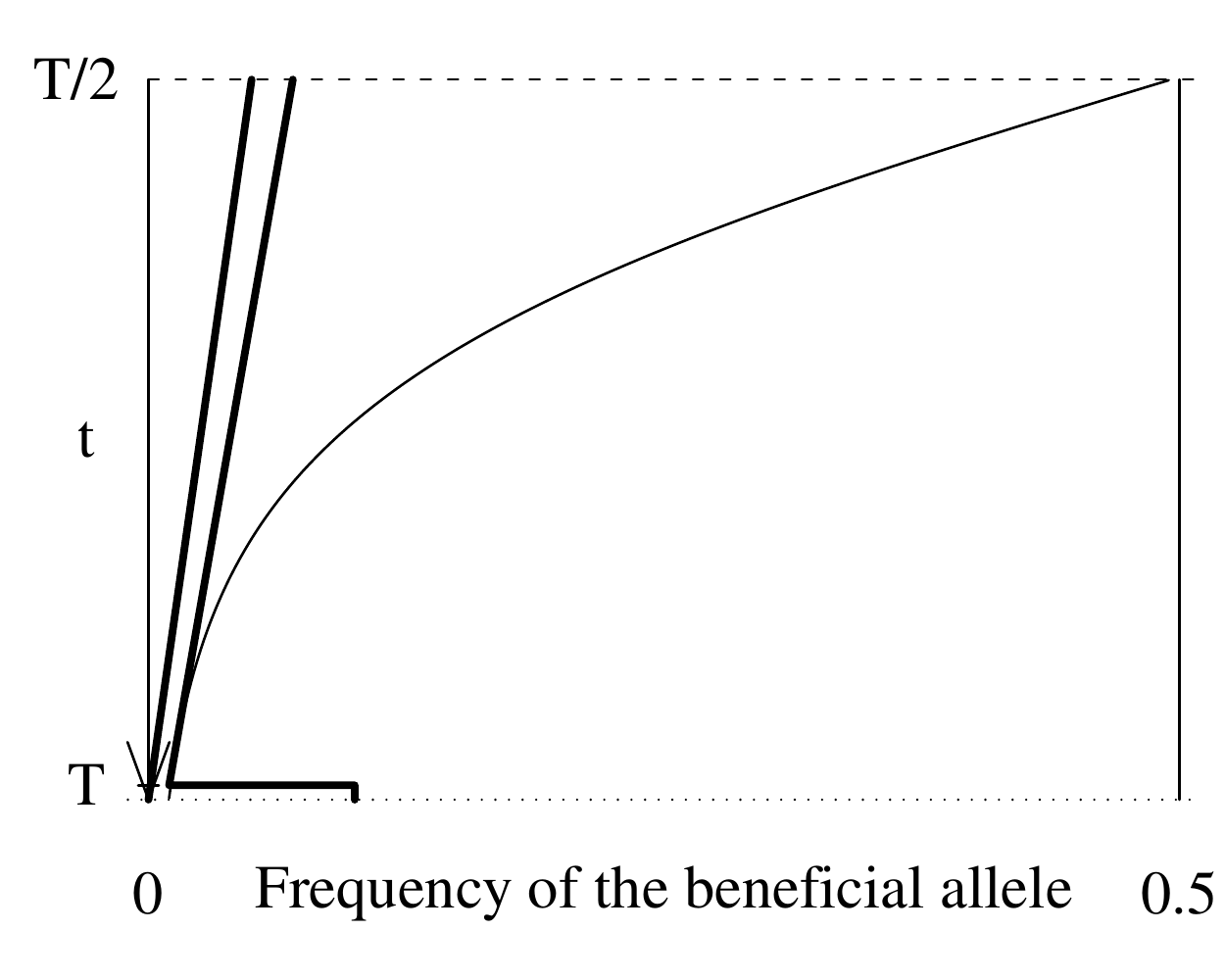}
             }
   \caption{The figure on the left shows two lines which coalesce before they mutate. The figure on the right shows two lines which mutate before they coalesce (further back in the past).}
    \label{fig:CoalMut1}
\end{figure}

\subsubsection*{Coalescent and mutations to the beneficial allele}
\begin{itemize}
 \item The rate of coalescence of two lines at time $t$ under the condition, that the two lines are in the beneficial background at time $t_{-}$ and the
 frequency of beneficial allele is $X_{t_{-}}= \frac{i}{2N}$ is equal to $$\frac{\frac{1}{(2N)(2N-1)} \frac{i(i-1)}{2N}}{ \frac{i(i-1)}{(2N)(2N-1)}} +
 \frac{\frac{1}{(2N)(2N-1)} \frac{(i-1)(2N-i+1)(1+s)}{2N}}{\frac{i(i-1)}{(2N)(2N-1)}}
= \frac{(1+s)(2N+1)}{2Ni} -\frac{ s}{2N} $$ for $i>1$. The parents of the beneficial offspring
are either both from the beneficial background or one is from the wild-type and the second from the beneficial background; for similar calculations see \citep{BartonEtheridgeSturm2004}, Lemma 2.4.
 For large $N$ this rate is approximately 
\begin{equation}\label{coalrate}
\frac{1+s}{2NX(t)},
\end{equation} since for large $N$ the
 frequency $X^{N}(t)$ is well 
approximated by the solution $X(t)$ of the differential Equation (\ref{diffeq}) by Proposition \ref{kurzthm}.

\begin{figure}
 \begin{center}
\includegraphics[scale=0.6]{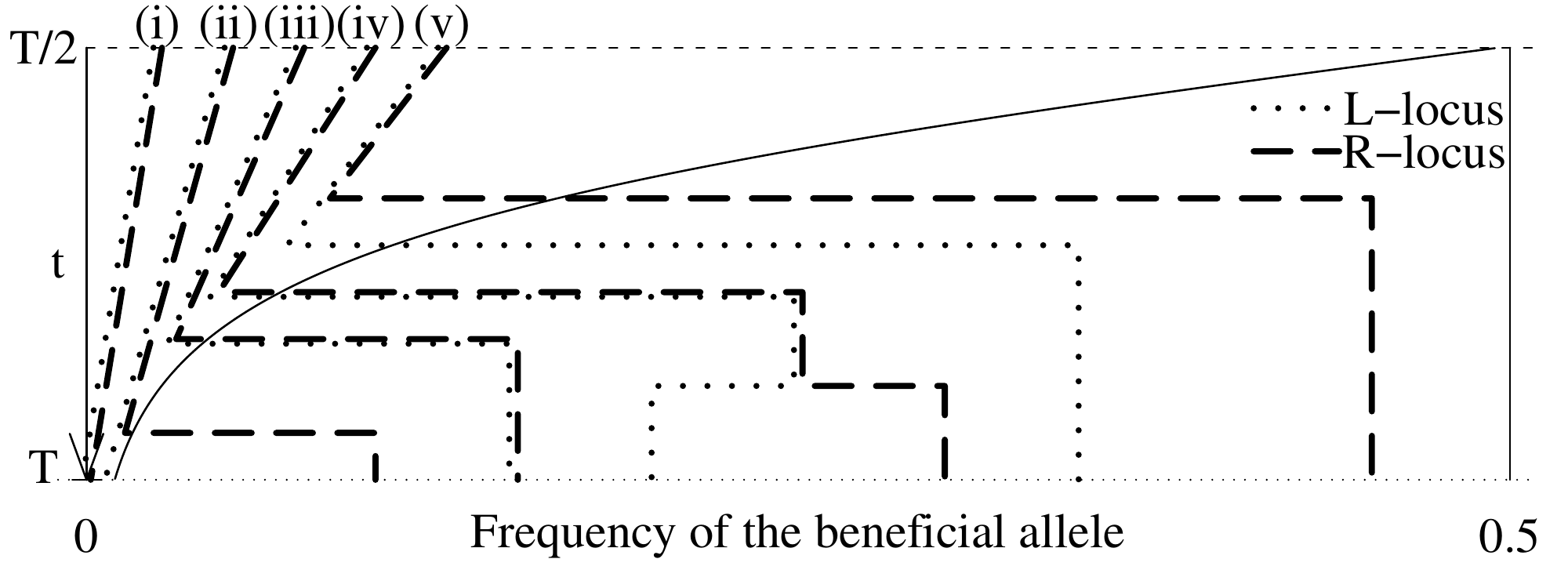} 
\end{center}
  \caption{\label{fig:starApproximation3}
Possible lines in the time interval $ [T/2,T)$ (backward in time)
 for geometry (a).
 The corresponding probabilities are given in Table
     \ref{tab:geoalpha}.
}
\end{figure}

\item If an individual mutates to the beneficial type, the genealogical line of this individual jumps (forward in time)
 from the wild-type background to the beneficial background. 
Backward in time the line is located at time $t_{-}$ in the beneficial background and at time $t$, after the mutation event, 
 in the wild-type background. 

The rate of mutation to the beneficial background of a line at time $t$ under the
 condition, that the frequency
of the beneficial allele is $X_{t_{-}}= \frac{i}{2N}$ at time $t_{-}$, is equal to
$$\frac{u_s \frac{2N-i}{2N}}{\frac{i}{2N}},$$
 
Analogous argumentations as for the coalescence rate yield, that this rate is approximately 
 \begin{equation}\label{mutrate}
\frac{u_s(1-X(t))}{ X(t)}\end{equation}
 for large N.
\end{itemize}

Both rates scale with $\frac{1}{X(t)}$, which means that the coalescent and mutation rates
 are high, if the frequency $X(t)$ is small.
 Hence it makes sense to assume that all mutations to the beneficial allele and all coalescent events occur at time $t=T$, i.e. at the beginning of the sweep. 
This scaling of the backward mutation rate shows that the star-like approximation, which has before been used for the classical
 hard sweeps case, should also be appropriate for soft sweeps.

With the approximate mutation and coalescent rates (\ref{mutrate}) and (\ref{coalrate}) the probability for a hard sweep in a sample of two can be bounded, see
also \citep{PenningsHermisson2006a} for a similar calculation in a Wright-Fisher-model formulation.
The probability for a hard sweep of two lines equals the probability, that the coalescent event happens before the mutation event.
 The mutation rate in a sample of two lines is approximately $2\frac{u_s(1-X(t))}{ X(t)} = \frac{\theta_s (1-X(t)}{2N X(t)}$,
 if terms of order $(2 N u_s)^2$ are ignored.

\begin{figure}
 \begin{center}
\includegraphics[scale=0.6]{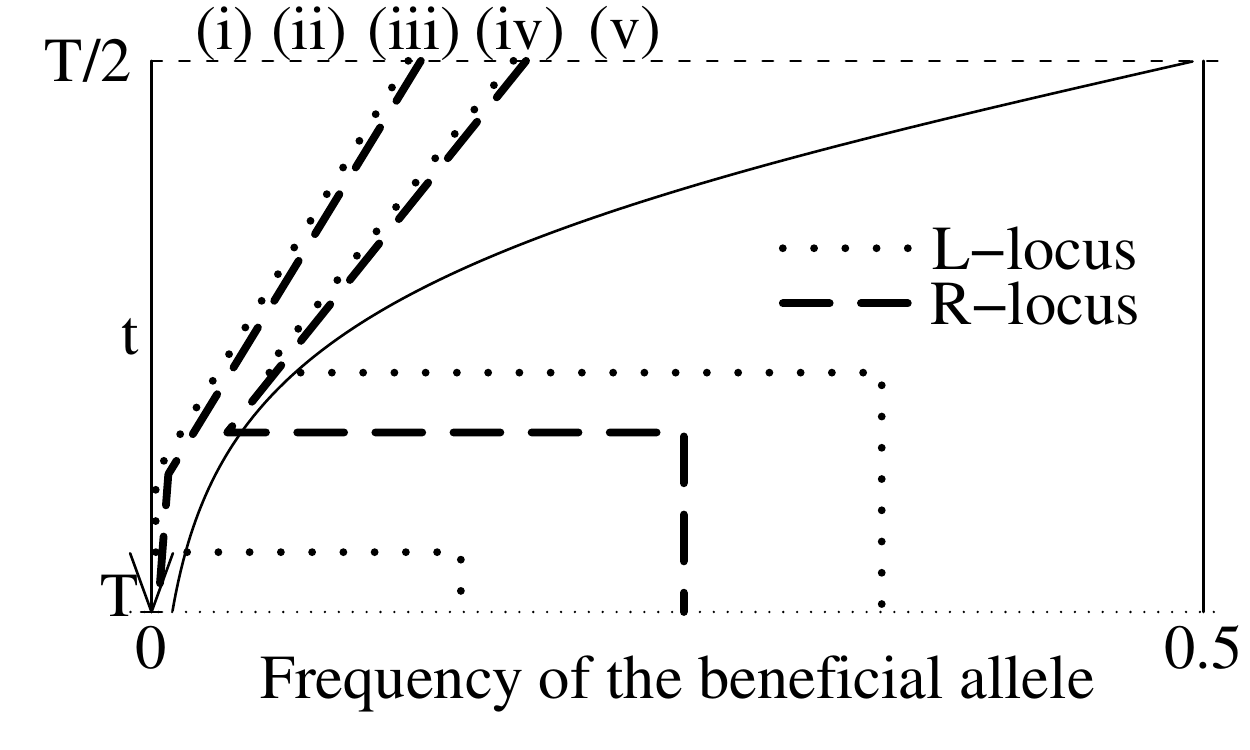} 
\end{center}
  \caption{\label{fig:starApproximation4}
Same as Figure
     \ref{fig:starApproximation3}
 for geometry (b). The lines (i), (ii) and (v) are not shown, as they are the same as in Figure  \ref{fig:starApproximation3}.
  The corresponding probabilities are given in Table \ref{tab:geobeta}.
}
\end{figure}

Let $C=(C_t)_{0\leq t \leq T}$, $M=(M_t)_{0 \leq t \leq T}$ be two independent Poisson processes with rates $\lambda(t)= \nolinebreak \frac{1+s}{2N X(T-t)},
 \mu(t) = \frac{ \theta_s(1-X(T-t))}{2N X(T-t)}$  respectively. Denote by
$S_1$ the first jump time of $C$ and by $T_1$ the first jump time of $M$.
Then the probability for a hard sweep in a sample of two is
$$P_{hard,2}:= P(S_1< T_1).$$ 
If the initial frequency of the beneficial allele $X(0)= \epsilon$ is small, 
\begin{align*}
P(S_1< T_1) & = \int \limits_0^{\infty} P(t<T_1) f_{S_1}(t) dt = 
\int \limits_{0}^{T} \exp\Big(-\int \limits_{0}^{t} \mu(\tau)d\tau\Big) \lambda(t) \exp\Big(-\int \limits_0^t \lambda(\tau) d\tau \Big) dt = \\
&=\int \limits_{0}^{T} \frac{1+s}{2N X(T-t)} \exp\Big(-\int \limits_{0}^{t} \frac{1+s+\theta_s(1-X(T-\tau))}{2N X(T-\tau)} d\tau \Big)dt= \\
& = \frac{1+s}{1+s+\theta_s} \Bigg( \int \limits_{0}^{T}\frac{1+s + \theta_s(1-X(T-t))}{2N X(T-t)}
 \exp\big(- \int \limits_{0}^{t} \frac{ 1+ s+\theta_s(1-X(T-\tau))}{2N X(T-\tau)} d\tau \big)dt\\
&\quad + \frac{\theta_s X(T-s)}{2N X(T-s)}\int \limits_{0}^{T} 
 \exp\big(-\int \limits_{0}^{t} \frac{ 1+s+\theta_s(1-X(T-\tau))}{2N X(T-\tau)} d\tau \big)dt \Bigg) \approx \\
& \approx \frac{1+s}{1+s+\theta_s}(1+ u_s \mathcal{T}),
\end{align*}
where $f_{S_1}(t)$ denotes the density of the probability measure induced by $S_1$ and
 $\mathcal{T}$ denotes the expected time for the first coalescent or mutation event. 
In the penultimate equation the first summand is the probability for a coalescence or a mutation event, this probability is approximately 1 for $\epsilon$ small.

\begin{table}
  \begin{center}
    \begin{tabular}{|c|c|c|}\hline 
      case & event & probability \\ \hline\hline
      \rule[-3mm]{0cm}{.7cm}(i) & \parbox{7cm}{no recombination event} & $p_{SL}p_{LR}$ \\\hline
      \rule[-8mm]{0cm}{1.8cm} (ii) & \parbox{7cm}{a $LR$-recombination event makes the allele at the $R$-locus escape the sweep without the allele at the $L$-locus} & $p_{SL}(1-p_{LR})$ \\\hline
      \rule[-8mm]{0cm}{1.8cm}  (iii) & \parbox{7cm}{by a $SL$-recombination event the line escapes the sweep and the alleles at the $L$- and $R$-locus stay linked} & $(1-p_{SL})p_{LR}$ \\\hline
      \rule[-22mm]{0cm}{4.8cm}  \parbox{.6cm}{(iv)\\[5ex](v)} & \parbox{7cm}{a $SL$-recombination event brings the alleles at the $L$- and $R$-loci linked into the wild-type background; here, the ancestry of both alleles is split by a $LR$-recombination\\[.5ex]a $LR$- and a $SL$-recombination event bring first the allele at the $R$-locus and then the allele at the $L$-locus into the wild-type background} & \parbox{3.8cm}{$\mathbb P[\text{(iv) or (v)}]$ \\ $\;\;=(1-p_{SL})(1-p_{LR})$} \\\hline
    \end{tabular}
  \end{center}
  \caption{\label{tab:geoalpha}Probabilities of several events happening 
    between times $T/2$ and $T_{-}$ for geometry (a); see Figure 
    \ref{fig:starApproximation3}. All events are described 
    backwards in time.} 
\end{table}

The time $\mathcal{T}$ lies approximately between 0 and the fixation time $T$ for $\epsilon$ small. So for $s\ll \theta_s$ we can (approximately) bound the probability for
 a hard sweep in a sample of two by
$$\frac{1}{1+\theta_s}\leq  P_{hard,2} \leq \frac{1}{1+\theta_s}(1+ u_s T)$$
Hence, the probability for a soft sweep in a sample of two can be (approximately) bounded by
$$\frac{\theta_s}{1+\theta_s} \geq   P_{soft,2} \geq \frac{\theta_s}{1+\theta_s}(1-u_s T) . $$

We can generalize this approach to obtain the (approximate) distribution of the number
of founders and offspring in a sample of size $n$.
It is given by Ewens sampling formula: 

In a sample of size $n$ at time 0, the probability, that there are $a_j$ founders of the sweep (with respect to the selected locus)
 which have $j$ offspring for $j \in \{1, ..., n\}$ is given by Ewens sampling formula
\begin{equation}\label{Ewenssamp}
 \frac{n!}{\theta_s(n)}  \prod \limits_{j=1}^{n} \frac{(\theta_s/j)^{a_j}}{a_j!} 
\end{equation}
where $\theta_s(n) = \theta_s \cdot (\theta_s+1) \cdot \dots \cdot (\theta_s+n-1).$
See \citep{PenningsHermisson2006b} for a derivation of this formula in a Wright-Fisher-model formulation.

We will assume in our approximation of the genealogy, that the number of founders and the number of their offspring is Ewens distributed as given in Equation 
(\ref{Ewenssamp}).

\begin{table}
  \begin{center}
    \begin{tabular}{|c|c|c|}\hline 
      line & event & probability \\ \hline\hline
      \rule[-3mm]{0cm}{.7cm}(i) & \parbox{7cm}{no recombination event} & $p_{LS}p_{SR}$ \\\hline
      \rule[-8mm]{0cm}{1.8cm} (ii) & \parbox{7cm}{a $SR$-recombination event makes the allele at the $R$-locus escape the sweep without the allele at the $L$-locus} & $p_{LS}(1-p_{SR})$ \\\hline
      \rule[-8mm]{0cm}{1.8cm}  (iii) & \parbox{7cm}{a $LS$-recombination event makes the allele at the $L$-locus escape the sweep without the allele at the $R$-locus} & $(1-p_{LS})p_{SR}$ \\\hline
      \rule[-17mm]{0cm}{3.5cm}  \parbox{.6cm}{(iv)\\[7ex](v)} & \parbox{7cm}{a $LS$-recombination event followed by a $SR$-recombination event bring the alleles at the $L$- and $R$-locus into the wild-type background\\[.5ex]same as (iv) but in reverse order of the $LS$- and $SR$-recombination events} & \parbox{3.8cm}{$\mathbb P[\text{(iv) or (v)}]$ \\ $\;\;=(1-p_{LS})(1-p_{SR})$} \\\hline
    \end{tabular}
  \end{center}
  \caption{\label{tab:geobeta}Probabilities for events happening between time $T/2$ and $T_{-}$ for geometry (b); see Figure 
    \ref{fig:starApproximation4}. All events backward in time} 
\end{table}
\subsubsection*{ Recombination events}
Forward in time at a recombination event two lines merge into one. If a recombination event occurs
 between two neighboring loci $L_1$ and $L_2$ (we will write $L_1 L_2$-recombination event, for short), such that $L_1$ lies on the left side of $L_2$, the
 offspring carries
at all loci left of the locus $L_1$ including the locus $L_1$ the alleles of the first parent and at the remaining loci the alleles of the
 second parent (with $L_1,L_2 \in \{L, R, S\}$).
Backward in time at a recombination event one line splits up into two lines.

Since coalescence events occur at the beginning of the sweep,
 one can assume, that each
 recombination event affects only a single line. The probability for no recombination event in the time interval $[t_1,t_2]$
is given by the probability, that the first jump time of a Poisson process started at time $t_1$ with rate $r(t)$ does not occur until
time $t_2$. The rate $r(t)$ depends on the different kinds of recombination events and is specified in the following.  
\begin{itemize}
\item The frequency $X_t$ stays between backward time $0$ and $T/2$ almost the whole time near 1 and is certainly greater 1/2. (The larger $\alpha$
the longer $X_t$ remains in a small neighborhood of 1.)
So, in the first half, recombination between the backgrounds is not frequent.
 Furthermore if $L$, $R$ and $S$ are arranged according to geometry
(a) $SL$-recombination events inside the beneficial background cannot be
seen in the DNA-data. The only events that can be recognized in the DNA-data and occur at a non negligible amount are $LR$-recombination
 events in the beneficial 
background.
If the loci are arranged according to geometry (b) all recombination events in the beneficial locus may be seen in the data.

The rate of recombination events between loci $L_1$ and $ L_2$ (with $L_1, L_2 \in \{L,R,S\}$ for geometry (b) and $L_1=L$ and $L_2=R$ for geometry (a))
 in the beneficial background is approximately $\frac{r_{L_1 L_2} X(T-t) X(T-t)}{X(T-t)}$ with $r_{L_1 L_2} \geq 0$.
Therefore the probability for no $L_1 L_2$-recombination event is given by
\begin{equation}\label{prob}
\exp \Big(-\int_{0}^{T/2} r_{L_1 L_2} X(T-t) dt \Big). 
\end{equation}
 As long as $u_s$ is small, the differential Equation (\ref{diffeq}) is only a small perturbation of the differential equation
$  \dot{X}(t) = \nolinebreak[10] s X(t)(1-X(t)).$  For this equation the integral in Equation (\ref{prob}) is equal
 to $ r_{L_1 L_2} \log(\alpha)/s + r_{L_1 L_2} \log(2)/s.$
Since the second summand is small, we approximate (\ref{prob}) by $$p_{L_1 L_2}:= \exp(- \rho_{L_1 L_2} \log(\alpha)/\alpha), $$ 
where $\rho_{L_1 L_2}:= 2N r_{L_1 L_2}$ denotes the recombination rate between the locus $L_1$ and $L_2$.

\begin{figure}
 \begin{center}
\includegraphics[scale=0.5]{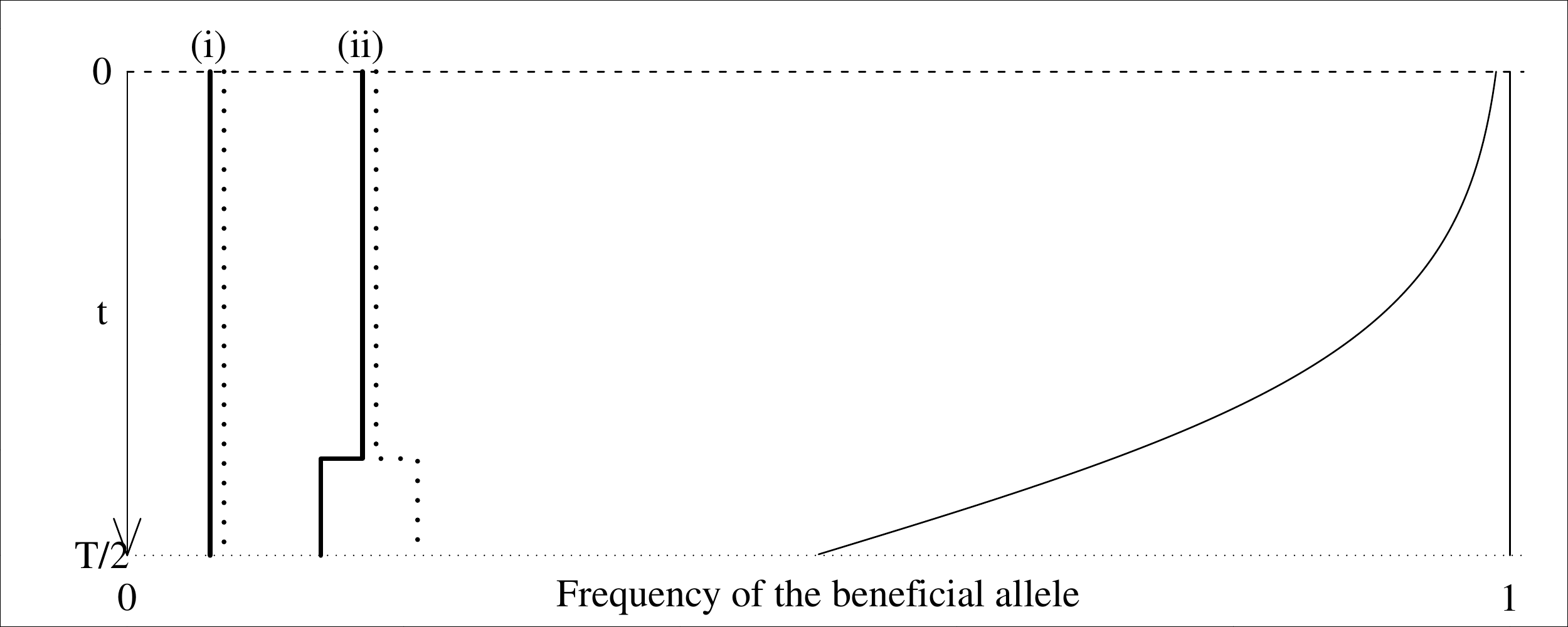} 
\end{center}
  \caption{\label{fig:ZeitintervallI3}
Possible split of two linked neutral loci: Two alleles at the neutral loci linked to the beneficial allele either (i) 
have a common ancestor at time $T/2$ or (ii) have two different ancestors that are both linked to a beneficial allele.}
\end{figure} 

\item In the time interval $[T/2,T]$ all recombination events with offspring in the beneficial background
 except recombination events inside the beneficial locus are probable to be seen in the data. Similar
arguments as above lead to the following assumption:
The probability for no recombination between locus $L_1$ and $L_2$ in the time interval $[T/2,T]$ is given by 
$$p_{L_1 L_2} := \exp(- \rho_{L_1 L_2} \log(\alpha)/\alpha). $$

\end{itemize}

See Figure \ref{fig:starApproximation3} - \ref{fig:ZeitintervallI3}
 for an illustration of the different types of recombination events possible in the time intervals $[0,T/2)$ and $[T/2,T)$. 

With this motivation, we can define an extended star-like genealogy:
 
\subsection{Genealogies: Definition} \label{GeneaDefi}

The joint genealogy of two neutral loci in the neighborhood of the selected locus can be defined as a structured partition-valued process.
Denote by $\Xi_A:= \{ \xi_A | \xi_A \textrm{ partition of } A \}$ the set of partitions of a set
$A$.
A partition $\xi_A = \{\xi_1, ..., \xi_m\} $ is called finer than a partition $\xi^{'}_A =\{\xi^{'}_1, ..., \xi^{'}_k\}$, iff
for each $j\in\{1, ..., m\}$ exists an $i \in\{1, ..., k\}$, such that $\xi_{j} \subseteq \xi^{'}_i$. 
We write  $\xi_A \preccurlyeq \xi^{'}_A$, if $\xi_A$ is finer than $\xi^{'}_A$.
Let $\xi= \{\xi_1, ..., \xi_m\} \in \Xi_A$ and $\eta= \{\eta_1, ..., \eta_k\}  \in \Xi_{\{1, ..., m\}}$, $m\geq k$, 
then the composition $$\eta \circ \xi := \bigg\{ \bigcup_{j\in \eta_1} \xi_j, ..., \bigcup_{j \in \eta_k} \xi_j \bigg\}.$$
A structured partition of $A$ is a tuple $(\xi_A^{B}, \xi_A^{b})$
 with $\{\xi^{B}_A \cup \xi^{b}_A\} \in  \Xi_A$ and $\xi_A^{B} \cap \xi_A^{b}= \emptyset$.
Partition elements in $\xi^{B}_A$ ($\xi^{b}_A$) are called beneficial (wild-type). 
 Denote by $$\Xi^{B,b}_A:= \{ (\xi_A^{B}, \xi_A^{b}) | (\xi_A^{B}, \xi_A^{b}) \textrm{ is a structred partition of } A \}$$
the set of structured partitions of the set A. Elements of a structured partition $(\xi_A^{B}, \xi_A^{b})$ are of the form $(\xi_1, \xi_2)$ with 
$\xi_1 \in \xi_A^{B}$ and $\xi_2\in \xi_A^{b}$.

Define $\dickm \ell:= \{1, ...., n \}$ the set of the L-loci and $\dickm r:= \nolinebreak[10] \{n+1, ..., 2n\}$ the set of the R-loci of a sample of size $n$ of the population.
We are interested in the structured partitions of $\dickm \ell \cup \dickm r$.

For geometry (a) the different kinds of recombination events can change the structured partition $(\xi^{B}, \xi^{b})$ to
\begin{itemize}
 \item $\left(\xi^B\setminus\{\xi^B_j\}, \xi^b\cup \{\xi^B_j\}\right),$ if an $SL$-recombination event happens between a wild-type and beneficial line
 and the offspring carries at the selected locus the beneficial allele (thus at the $L$- and $R$-locus the individual carries the alleles of the wild-type line). \\
 \item $\left((\xi^B\setminus\{\xi^B_j\}) \cup \{\xi^B_j\cap \dickm \ell\}, \xi^b\cup \{\xi^B_j\cap\dickm r\} \right)$, if an $LR$-recombination event happens
between an individual of the beneficial background and an individual of the wild-type background and at the $S$ and $L$-locus the beneficial line is carried
 on (forward in time).\\
 \item $\left((\xi^B\setminus \{\xi^B_j\}) \cup \{\xi^B_j\cap \dickm \ell, \xi^B_j\cap \dickm r\}, \xi^b\}\right)$, if an $LR$-recombination event
 happens between two individuals of the beneficial background.\\
  \item $\left(\xi^B, (\xi^b \setminus\{\xi^b_k\}) \cup \{\xi^b_k\cap \dickm \ell, \xi^b_k\cap\dickm r\}\right),$ if an $LR$-recombination event happens between
 two individuals of the wild-type background.\\
 \item  $\left(\xi^B\cup \{\xi^b_k\}, \xi^b\setminus\{\xi^b_k\}\right),$ if an $SL$-recombination event happens between a beneficial and a wild-type line
 and the offspring carries the beneficial allele.\\
 \item $\left(\xi^B\cup \{\xi^b_k\cap\dickm r\}, (\xi^b\setminus\{\xi^b_k\}) \cup \{\xi^b_k\cap \dickm \ell\}\right),$ if an $LR$-recombination event happens 
between a beneficial and wild-type line and at the selected and $L$-locus the wild-type is carried on (forward in time).
\end{itemize}

For geometry (b) the partitions change in an analogous manner.

Before we give the definition of an extended star-like genealogy we define genealogies and  samples.
\begin{defi}
The genealogy of a set $A$ is a four-time step Markov chain
$$(\xi_t)_{t\in\{0,T/2,T_{-},T\}} =((\xi^{B}_{0},\xi^{b}_{0}), (\xi^{B}_{T/2}, \xi^{b}_{T/2}),(\xi^{B}_{T_{-}},\xi^{B}_{T_{-}}), (\xi^{B}_{T}, \xi^{b}_{T}))$$ with
 state space $\Xi^{B,b}_{A}$.

A set $\dickm \ell \cup \dickm r$ with $\dickm \ell:= \{1,..., n\}$ and $\dickm r:=\{n+1, ..., 2n\}$, $n\in \mathbbm{N}$, is a
 sample at two loci $L$ and $R$ taken from the population at time $t=0$, if
the genealogy of the sample is at time $t=0$ given by $\xi_0= \big(\big\{\{1, n+1\}, \dots, \{n, 2n\}\big\}, \{\emptyset\} \big)$.
\end{defi}

With this we can define an extended star-like genealogy as a four-time step random experiment: 
\begin{defi}\label{stargene}
An extended star-like genealogy of a sample $\dickm \ell \cup \dickm r$ with $\dickm \ell:= \{1,..., n\}$ and $\dickm r:=\{n+1, ..., 2n\}$
 at two loci $L$ and $R$ in the neighborhood of a selected locus $S$
 arranged according to geometry (a) (resp. geometry (b)) is a four-time step Markov chain
 $$(\xi_t)_{t\in\{0,T/2,T_{-},T\}} =((\xi^{B}_{0},\xi^{b}_{0}), (\xi^{B}_{T/2}, \xi^{b}_{T/2}),(\xi^{B}_{T_{-}},\xi^{B}_{T_{-}}), (\xi^{B}_{T}, \xi^{b}_{T}))$$ with
 state space $\Xi^{B,b}_{\dickm \ell \cup \dickm r}$ with the following properties:
\begin{itemize}
\item At time $T/2$:
   \begin{itemize}
    \item Structured partition elements are stochastically independent
    \item No recombination events between the backgrounds and no mutations to the beneficial allele, i.e. $P\big( \xi^{b}_{T/2}=\emptyset\big)=1$
    \item No coalescence events, i.e. $P\big(\xi^{B}_{T/2}  \preccurlyeq \xi^{B}_{0}  \big)=1$  
    \item For geometry (a): No $LR$-recombination events in the beneficial background with probability $p_{LR}$, i.e. for $j \in \dickm l$
          a structured partition element at time $t=0$ of the form $(\{\{j, j+n\}\}, \{\emptyset\})$
          \begin{itemize}
          \item is kept at time $T/2$ with probability $p_{LR}$
          \item  and changed to $(\{\{j\}, \{j+n\}\}, \{\emptyset\})$ with probability $1-p_{LR}$
          \end{itemize}
    \item For geometry (b): Neither a $LS$-recombination events nor a $SR$-recombination events in the beneficial background happens
          with probability $p_{LS}p_{SR}$, i.e. for $j \in \dickm l$
          a structured partition element at time $t=0$ of the form $(\{\{j, j+n\}\}, \{\emptyset\})$
         \begin{itemize}
         \item is kept till time $T/2$ with probability $p_{LS}p_{SR}$
         \item and changed to $(\{\{j\}, \{j+n\}\}, \{\emptyset\})$ with probability $1-p_{LS}p_{SR}$ 
         \end{itemize}       
\end{itemize}
\item At time $T_{-}$:
\begin{itemize}
 \item Structured partition elements are stochastically independent 
 \item No coalescence events, i.e. $P(\xi^{B}_{T_{-}} \cup \xi^{b}_{T_{-}}  \preccurlyeq \xi^{B}_{T/2}  )=1$ 
 \item For geometry (a): For $j \in \dickm l$ 
\begin{itemize} 
\item a partition element at time $T/2$ of the form $(\{\{j, j+n\}\}, \{\emptyset\})$ \\
         is kept at time $T_{-}$ with probability $p_{SL}p_{LR}$, \\
         changed to $( \{\emptyset\}, \{ \{j, j+n\}\})$ with probability $(1-p_{SL})p_{LR}$, \\ 
           changed to $(\{ \{j\}\} , \{\{j+n\}\})$ with probability $p_{SL}(1-p_{LR})$ \\
          and changed to $(\emptyset, \{\{j\}, \{j+n\}\})$ with probability $(1-p_{SL})(1-p_{LR}).$ 
\item a partition element at time $T/2$ of the form $(\{j\}, \{\emptyset\})$ \\
             is kept at time $T_{-}$ with probability $p_{SL}$ \\
      and changed to $(\{\emptyset\}, \{j\})$ with probability $1-p_{SL}.$
\item a partition element at time $T/2$ of the form $(\{j+n\}, \{\emptyset\})$ \\
             is kept at time $T_{-}$ with probability $p_{SR}$ \\
      and changed to $(\{\emptyset\}, \{j+n\})$ with probability $1-p_{SR}.$
\end{itemize}
\item For geometry (b):
\begin{itemize} 
\item A partition element at time $T/2$ of the form $(\{\{j, j+n\}\}, \{\emptyset\})$ \\
         is kept at time $T_{-}$ with probability $p_{LS}p_{SR}$, \\ 
            changed to $(\emptyset, \{\{j\}, \{j+n\}\})$ with probability $(1-p_{LS})(1-p_{SR}).$ \\
            changed to $(\{\{j\}\} , \{\{j+n\}\})$ with probability $p_{SL}(1-p_{SR})$ \\
          and changed to $(\{\{j+n\}\}, \{\{j\}\})$ with probability $(1-p_{SL})p_{SR}.$ 
\item A partition element at time $T/2$ of the form $(\{j\}, \{\emptyset\})$ \\
             is kept at time $T_{-}$ with probability $p_{LS}$ \\
      and changed to $(\{\emptyset\}, \{j\})$ with probability $1-p_{LS}.$
\item A partition element at time $T/2$ of the form $(\{j+n\}, \{\emptyset\})$ \\
             is kept at time $T_{-}$ with probability $p_{SR}$ \\
      and changed to $(\{\emptyset\}, \{j+n\})$ with probability $1-p_{SR}.$
\end{itemize}
\end{itemize} 
\item At time $t=T$: \\
 At the beginning of the sweep all coalescence and mutation events happen: 
Let $ m \in \mathbbm{N}$ and $ a_j\in \{0, ..., m\}$ with $\sum_{j=1}^{m}j a_j=m.$ Denote by 
$M^{m} :=\{(\xi^B, \xi^b) \in \Xi^{B,b}_{\dickm \ell \cup \dickm r } ; |\xi^B|=m \}$ the set of structured partition of $\dickm \ell \cup \dickm r$
which beneficial partitions consist of m elements and by $N^{(a_1, ..., a_m)}:= \{ \eta= (\eta_1,..., \eta_k) \in \Xi_{\{1, ..., m\}}; 
\# \{\eta_l ; |\eta_l|=j\}=a_j \}$ the set of partitions of $\{1, ..., m\}$ containing $a_j$ partition elements of size $j$.
Then for  $\eta\in N^{(a_1, ..., a_m)}$
\begin{equation}\label{probfounder}
P( \xi_T= (\{\emptyset\}, (\eta \circ \xi^{B}) \cup \xi^{b}) | \xi_{T_{-}}=(\xi^{B}, \xi^{b}) \in M^{m}) = \frac{m!}{\theta_s(m)} \prod \limits_{j=1}^{m} \frac{(\theta_s/j)^{a_j}}{a_j!}.
\end{equation}
\end{itemize}
We say, that a population evolved according to an extended star-like genealogy, if the genealogy of each sample of the population is extended star-like. 
\end{defi}

\begin{bem}
\emph{ 
If we are interested in the genealogy of a subset $M$ of a sample $\dickm \ell \cup \dickm r$, the
 genealogy of $M$
 fulfills 
all conditions of Definition \ref{stargene}.
In particular, at time $t=T$ the number of founders together with the number of their offspring is Ewens distributed, since Ewens sampling formula is consistent.
At time $t=0$, the genealogy of $M$ is given by $$\xi_0=  \big(\big\{\{1, n+1\}\cap M, \dots, \{n, 2n\} \cap M\big\}, \{\emptyset\} \big).$$
}
\end{bem}
In accordance to the possible recombination events during the time interval $[T/2,T)$ we obtain the ancestral lines 
shown in Figure \ref{fig:starApproximation3} for geometry (a) and Figure
\ref{fig:starApproximation4} for geometry (b). The probabilities for these events are listed in Table \ref{tab:geoalpha} for geometry
 (a), in Table \ref{tab:geobeta} for geometry (b). For the time interval $[0,T/2)$ the possible ancestral lines are
 shown in Figure \ref{fig:ZeitintervallI3}. In Figure \ref{fig:CoalMut1} the left picture shows two lines which coalesce first and
mutate then, in the right picture the lines mutate first and coalesce afterwards.

\section{Results} 
Our main result is the computation of the linkage disequilibrium at the end of the sweep
 measured by $\mathbbm{E} [D_{\ell, r}(0) |D_{\ell, r}(T)]$ for two fixed allelic variants
$\ell$ and $r$ and $\widehat{\sigma^2_D}$ for two neutral 
loci in a neighborhood of the selected locus (backward in time). 

We apply the procedure of \citet{PfaffelhuberLehnertStephan2008} to compute $\mathbbm{E}[D_{\ell, r}(0) |D_{\ell, r}(T)]$. The main
 difference between our model and the hard sweep model is, that two 
lines do not have to coalesce, since both lines may mutate to the beneficial allele.
\begin{thm}\label{smallthm}
Assume, that the population evolved in a DNA-region containing the two neutral loci $L$ and $R$ and the selected locus $S$ according to an extended star-like
 genealogy and both loci carry exactly two allelic variants $\ell/L$ and $r/R$. Then the linkage 
disequilibrium of the allelic variants $\ell$ and $r$ measured by $\mathbbm{E}[D_{\ell,r}(0)|D_{\ell, r}(T) ]$ at the end of the sweep
is given by
\begin{equation}\label{Geoa}
  \mathbbm{E}[D_{\ell, r}(0) |D_{\ell, r}(T)] = p_{LR}^2(1-\frac{1}{1+\theta_s} p_{SL}^2)D_{\ell,r}(T), 
\end{equation}
if the two neutral loci are arranged according to geometry (a)
and 
\begin{equation}\label{Geob}
 \mathbbm{E}[D_{\ell, r}(0) |D_{\ell, r}(T)] = p_{LR}^2(1-\frac{1}{1+\theta_s})D_{\ell,r}(T),
\end{equation}
for L and R arranged according to geometry (b). 
\end{thm}
\begin{proof}
Indeed, consider the genealogy $(\xi)_{t\in\{0,T/2,T_{-},T\}}$ of an $L$-locus $\{1\}$ and an $R$-locus $\{2\}$, i.e. $\xi_t \in \Xi^{B,b}_{\{1,2\}}$. 
Denote by $d$ the probability, that the pair $\{1\},\{2\}$
 was linked at the beginning of the sweep, if it is linked at the end of the sweep, i.e.~let $d:= P(\xi_T=(\{\emptyset\}, \{1,2\})| \xi_0=\nolinebreak[30] (\{\{1,2\}\}, \{\emptyset\})).$
   Analogously, denote by $e$ the probability, that
the pair has been linked at the beginning, if it is unlinked at the end of the sweep.
 I.e.~$e:= P(\xi_T=(\{\emptyset\},\{\{1,2\}\})| \xi_0=(\{\{1\},\{2\}\},\{\emptyset\}))$. Then we can write  
\begin{equation*}
 \mathbbm{E}[q_{\ell r}(0)| q_{\ell r}(T), q_\ell(T),q_r(T)] = d q_{\ell r}(T) + (1-d)q_\ell(T)q_r(T)
\end{equation*}
\begin{equation*}
 \mathbbm{E}[q_{\ell}(0) q_{r}(0)| q_{\ell r}(T), q_\ell(T), q_r(T)] = e q_{\ell r}(T) + (1-e) q_\ell(T)q_r(T)
\end{equation*}
with $q_{\ell r}\in [0,1]$ and $q_r,q_\ell \in (0,1)$
and so
\begin{equation*}
 \mathbbm{E}[D_{\ell, r}(0) |D_{\ell, r}(T)=x] = (d-e) D_{\ell, r}(T).
\end{equation*}
The probabilities $d$ and $e$ are for geometry (a) and (b) given by
\begin{equation*} a) \ e = \frac{1}{1+ \theta_s} p_{SL}p_{SR} = \frac{1}{1+\theta_s}  p_{SL}^2 p_{LR} \quad  b) \ e= \frac{1}{1+ \theta_s} p_{LS}p_{SR}= \frac{1}{1+ \theta_s} p_{LR}
 \end{equation*}
and 
\begin{equation*} a) \ d = e(1- p_{LR}) + p_{LR}p_{LR} \quad b) \ d= e(1- p_{LR}) + p_{LR}p_{LS}p_{SR} =  p_{LR}.
 \end{equation*} 
In words, a pair is unlinked at the end of the sweep when it was linked at the beginning, iff
 the pair just coalesces, i.e.
 neither a recombination event between the $S$ and the $L$ locus neither 
a recombination event between the $S$ and $R$ locus occurred and the two loci coalesced before they mutated. 
And a pair which is linked at the end of the sweep is also linked in the beginning, iff either nothing happens or the pair is divided by a LR-recombination
event first and then linked again by coalescence event. 

From this easily follows Equation (\ref{Geoa}) for geometry (a) and Equation (\ref{Geob}) for geometry (b).
\end{proof}

To compute the quantity $\widehat{\sigma^2_D}$ consider the three quantities:

\begin{equation} \label{XYZ}
\begin{aligned}
 \mathcal{X}_t & := \mathbbm{E}[q_L(t)(1-q_L(t))q_R(t)(1-q_R(t))] \\ 
 \mathcal{Y}_t &:= \mathbbm{E}[D(t)(1-2q_L(t))(1-2q_R(t))] \\
 \mathcal{Z}_t &:= \mathbbm{E}[(D(t))^2]
\end{aligned}
 \end{equation}
for $0\leq t \leq T$.

\begin{thm}\label{bigthm}
Given $\mathcal{X}_{T}, \mathcal{Y}_{T}$ and $\mathcal{Z}_{T}$ at the beginning of the sweep and a sample of size $n$ of a population at the end of the sweep, i.e.~a set of $L$-loci
$\dickm \ell:=\{1, ...,n\}$ and a set of $R$-loci $\dickm r:= \{n+1, ..., 2n\}$, assume
 that the genealogy of the sample is extended star-like. Then the standard linkage disequilibrium $\widehat{\sigma^2_D}$ of this sample of
 two neutral loci at the end of a sweep equals:
\begin{equation}\label{sigma}
\widehat{\sigma^2_D}= \widehat{\mathcal{Z}_0}/\widehat{\mathcal{X}_0}.
\end{equation}
with
\begin{equation} \label{big1} 
\begin{aligned} 
\widehat{ \mathcal{Z}_0}& =  p_{LR}^4 (p_{SL}-1)^2 \big( p_{SL}^2(\mathcal{X}_T+\mathcal{Y}_T)+(1+2p_{SL})\mathcal{Z}_T \big) \\
& \quad + \theta_s p_{SR}^2 \Big(\frac{\mathcal{X}_T}{3} \big( p_{LR}(11p_{SR}-2-6p_{LR})+p_{SR}(2-4p_{SR}) \big) \\
& \quad + \frac{\mathcal{Y}_T}{12} \big( p_{SR}(2-21p_{SR}) + p_{LR}(38p_{SR}-15p_{LR}-2) \big) + \frac{\mathcal{Z}_T}{3}\big( 9p_{LR}^2-9p_{SR}p_{LR}+p_{SR}^2 \big) \Big)  
\end{aligned}
\end{equation}
and
\begin{equation}\label{big2}
\begin{aligned}
\widehat{ \mathcal{X}_0} & = (1-p_{SR})(p_{SL}-1) \big(\mathcal{X}_T(1+p_{SR}+p_{SL})+(\mathcal{X}_T+\mathcal{Y}_T)(p_{SL}p_{SR})\big) \\
& \quad + \frac{\theta_s p_{SL}^2}{3}\Big(\mathcal{X}_T(3p_{LR}^2-5p_{LR}+3 +2p_{SR}p_{LR}+2p_{SR}-4p_{SR}^2)  \\
& \quad + 5\mathcal{Y}_T(p_{SR}p_{LR}+p_{SR}- \frac{17}{20}p_{LR})+  p_{SR}^2(\mathcal{Z}_T-\frac{21}{4}\mathcal{Y}_T)\Big) 
\end{aligned}
\end{equation}
 for geometry (a)
and with
\begin{equation}\label{big3}
\begin{aligned}
\widehat{\mathcal{Z}_0} & = \theta_s p_{LR}^3 \Big(\frac{\mathcal{X}_T}{3}\big(1-p_{LS}+2p_{LR}-p_{SR}\big) \\
& \quad + \frac{\mathcal{Y}_T}{12} \big(3p_{LR}-p_{SR}-p_{LS}+1 \big)+ \frac{\mathcal{Z}_T}{3}p_{LR} \Big) 
\end{aligned}
\end{equation}
and
\begin{equation}\label{big4}
\begin{aligned}
\widehat{\mathcal{X}_0}&=\mathcal{X}_T \big((1-p_{LS}^2)(1-p_{SR}^2) \big)+ \mathcal{Y}_T \big(p_{LR}(1-p_{LS})(1-p_{SR})\big)\\
& \quad + \theta_s \Big( \frac{\mathcal{X}_T}{3}\big(3p_{LS}^2+2p_{LR}p_{SR}-4p_{LR}^2-5p_{LR}+2p_{LR}p_{LS}+3p_{SR}^2 \big) \\
& \quad + \frac{\mathcal{Y}_T p_{LR}}{12} \big(20 p_{LS}-21p_{LR}-17+20 p_{SR} \big) + \frac{\mathcal{Z}_T}{3}p_{LR}^2 \Big) 
\end{aligned}
\end{equation} 
for geometry (b), 
if we ignore in both geometries terms of order $\theta_s^2$ and $1/n$.
\end{thm}
For a proof of this theorem see Section \ref{sec:proof}.

\begin{bem}\label{bemtheo}
\emph{
\begin{itemize}
 \item In the supporting online material you find a Mathematica-notebook for computing the exact values of the standard linkage disequilibrium
measured by $\widehat{\sigma^2_D}$ without ignoring terms of order $\theta_s^2$ and $1/n$.
 \item If the population evolves neutrally till the beginning of the sweep, \citet{Ohta1969} have shown, that
\begin{align*}\label{XYZ_0}
\mathcal{X}_T & = \frac{1}{4} \frac{\theta^2}{\theta+1} \cdot  
  \frac{5+2\theta+\rho_{LR})(3+2\theta+2\rho_{LR})-4}{(1+\theta)(3+2\theta +2 \rho_{LR})(5+ 2\theta +\rho_{LR})- 2(3+ 2\theta)} \\
\mathcal{Y}_T & =  \frac{\theta^2}{\theta+1} \cdot 
 \frac{1}{(1+\theta)(3+2\theta +2 \rho_{LR})(5+ 2\theta +\rho_{LR})- 2(3+ 2\theta)} \\
\mathcal{Z}_T & = \frac{1}{4} \frac{\theta^2}{\theta+1} \cdot 
 \frac{2\theta +\rho_{LR} +5}{(1+\theta)(3+2\theta +2 \rho_{LR})(5+ 2\theta +\rho_{LR})- 2(3+ 2\theta)},  
\end{align*}
where $\theta:= 4Nu$ is the neutral mutation rate. For a comparison of the theoretical results with simulations
 we assume that the population evolved neutrally till the beginning of the sweep.
\item See Figure \ref{fig:Theo} for a plot of the theoretical values of $\widehat{\sigma^2_D}$ for different values of $\theta_s$.   
Here we assumed as well neutral evolution till the beginning of the sweep. 
 \item Note, that for $\theta_s=0$ we obtain $\widehat{\sigma^2_D}$ for a hard sweep, compare \citep{PfaffelhuberLehnertStephan2008}.
\end{itemize}
}
\end{bem}

\begin{figure}
 \begin{center}
\includegraphics[scale=0.35]{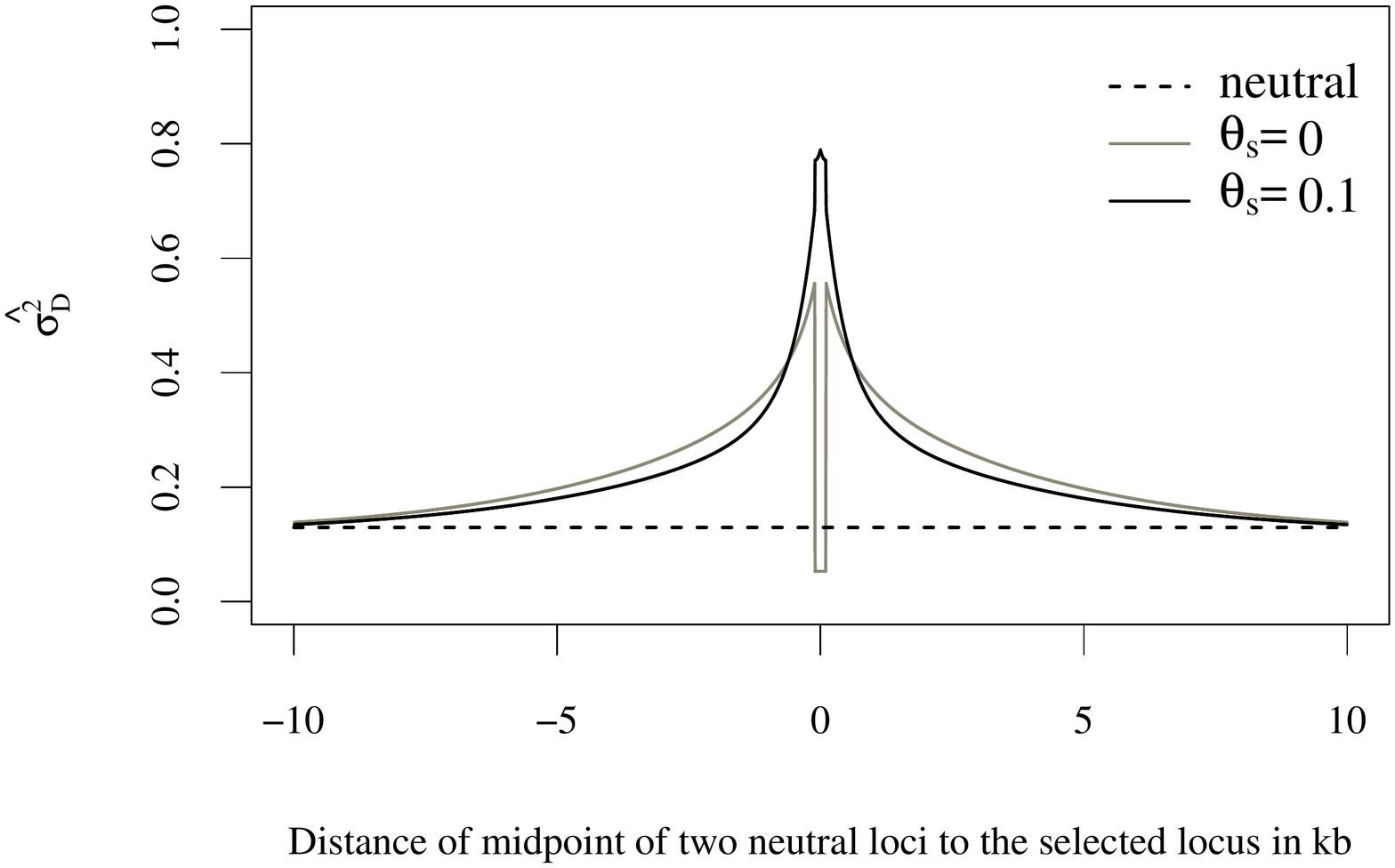} \end{center}
  \caption{\label{fig:Theo}
Theoretical values of $\widehat{\sigma^2_D}$ in the neutral setting, for $\theta_s=0$ and $\theta_s=0.1$.
 The distance between the neutral loci is $0.2$ kb, the selection strength $\alpha=1000$, the population size $N=10^6$, 
the recombination rate between the neutral loci $\rho_{LR}= 5$ and the neutral mutation rate $\theta=0.005$.
}
\end{figure}

\section{Simulations}\label{sim}
\begin{figure}
\subfigure{
\includegraphics[scale=0.35]{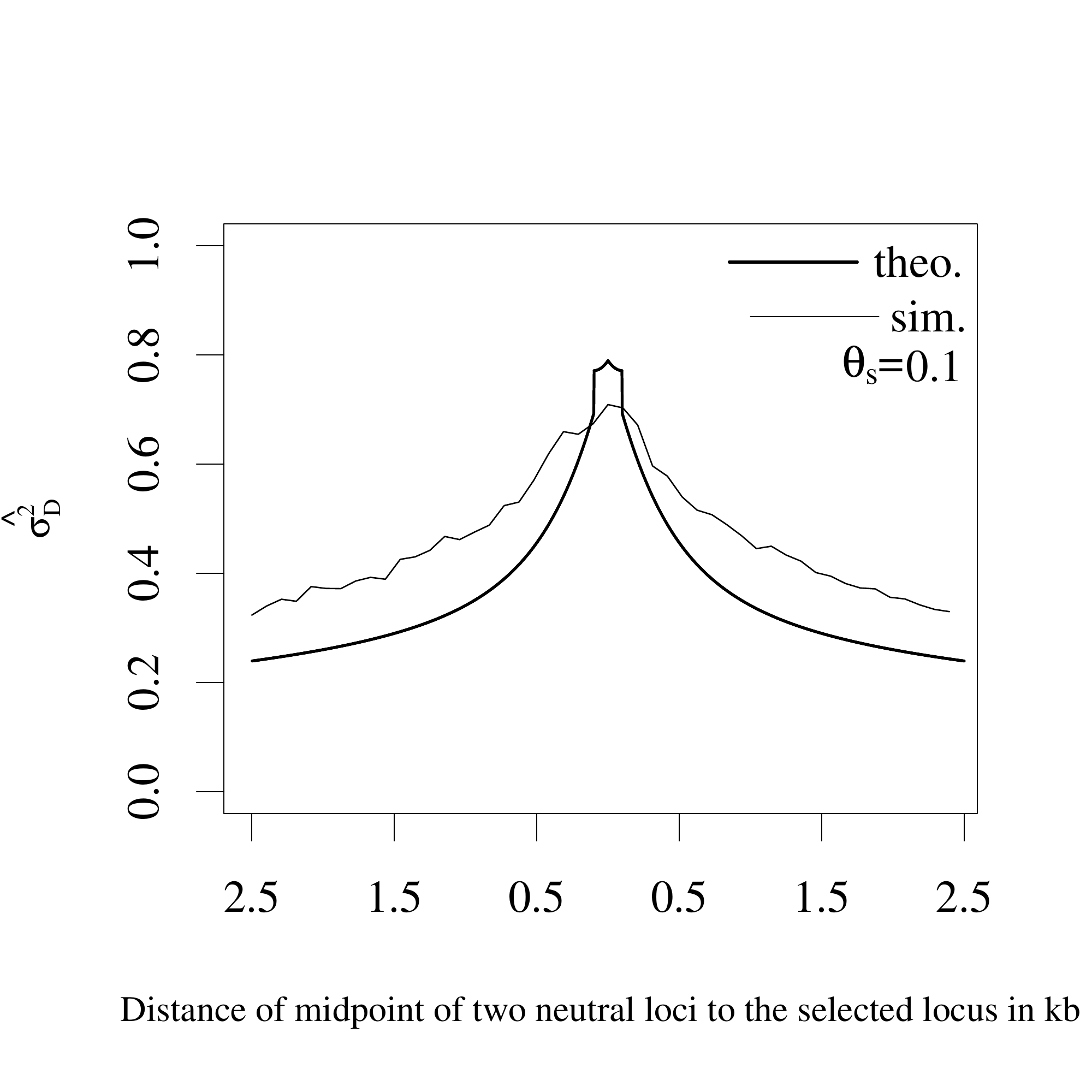}   
}
\subfigure{
\includegraphics[scale=0.35]{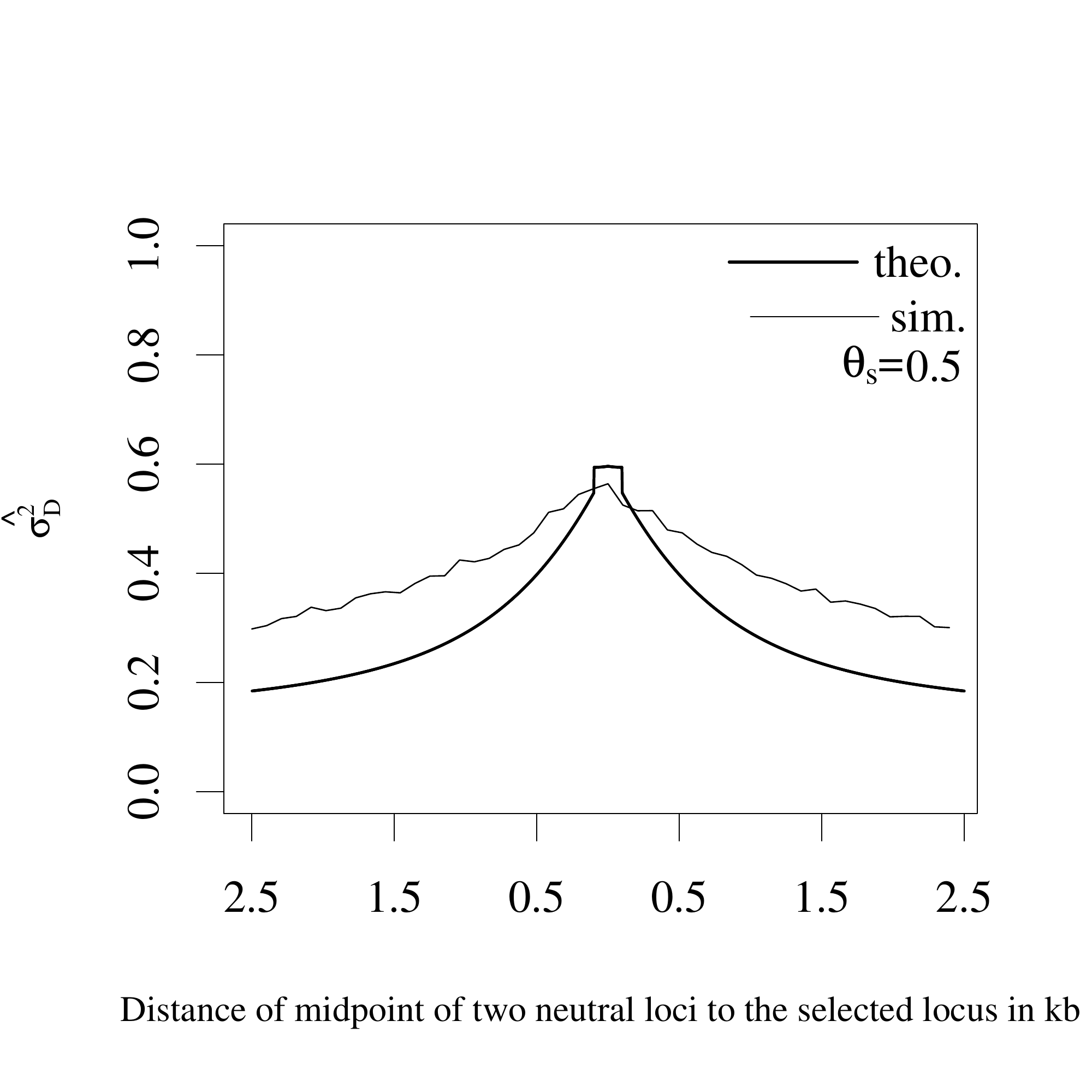}  
}
 \caption{Left figure: 
Plot of $\widehat{\sigma^2_D}$ for a neutral mutation rate $\theta= 0.005$, recombination rate $\rho= 0.025$, selection strength $\alpha=1000$, 
recurrent mutation rate to the beneficial allele $\theta_s=0.1$, a distance of 200 bp between the two neutral loci and a DNA-stretch of 
length 5 kb based on $10^4$ draws. Right figure:  
Plot of $\widehat{\sigma^2_D}$ with the same parameters as in the left figure except for the 
recurrent mutation rate to the beneficial allele $\theta_s=0.5$. 
}
\label{fig:Plot1}

\end{figure}

\begin{figure}
 \begin{center}
\includegraphics[scale=0.5]{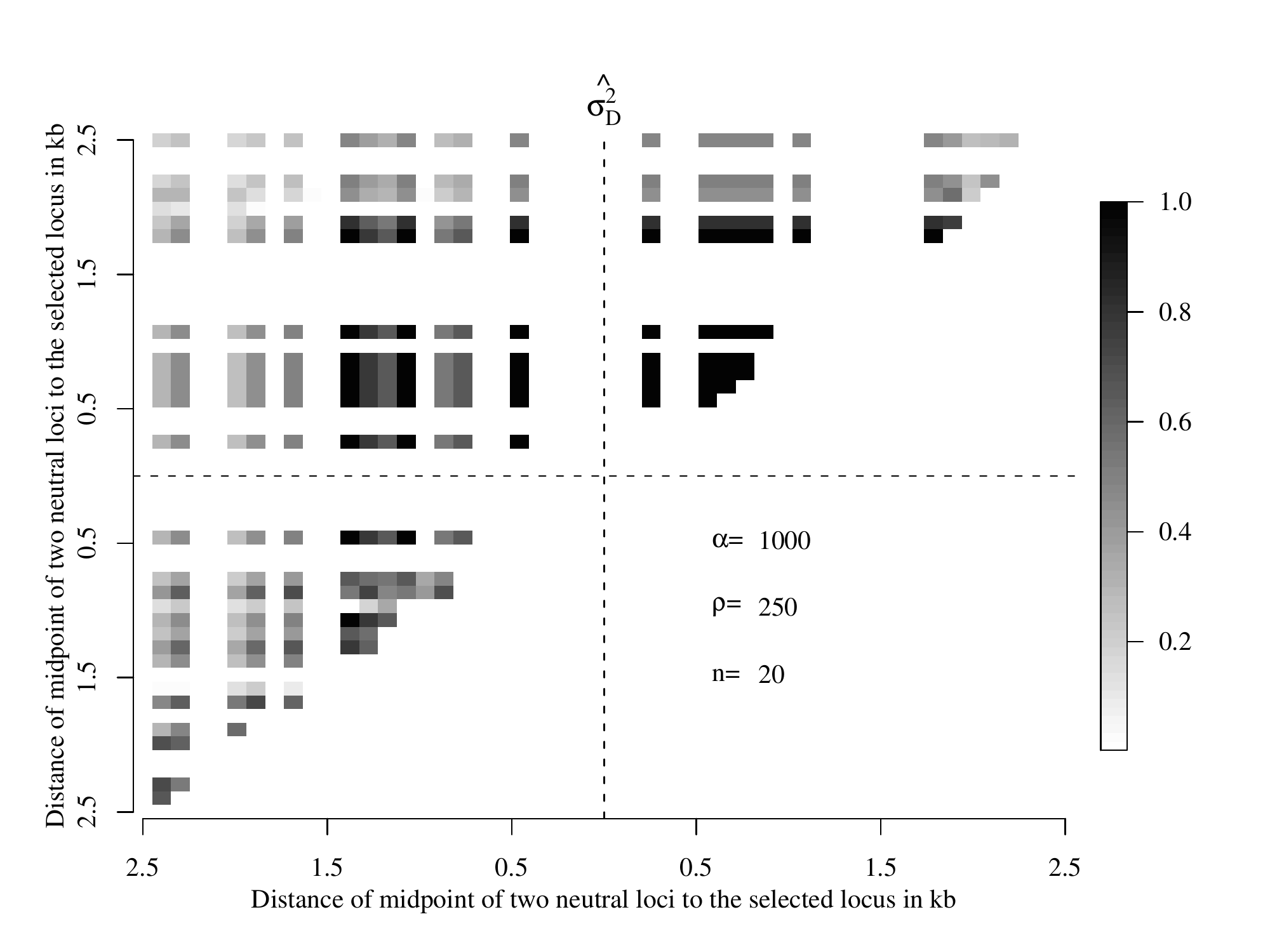} \end{center}
  \caption{\label{fig:Singlesample}
The full linkage disequilibrium spectrum for a single sample of a soft sweep with two founders with respect to the beneficial locus.
}
\end{figure}

We simulated sequence samples with the new program msms (for ms mit Selektion (German: with selection)) of Greg Ewing, see \citep{Ewing} to compare
our theoretical linkage disequilibrium values with linkage disequilibrium values obtained from simulated genealogies assuming neutral evolution till
 the beginning of the sweep.
The program msms generates sequence samples for a single selected locus of a population reproducing according to the Wright-Fisher-model 
 with the possibility of recurrent mutation to the beneficial allele. The frequency of the beneficial
allele is simulated stochastically conditioned on fixation. 
In Section 2 we argued, that the star-like genealogies approximates the Moran model genealogies well.
For large population sizes the Moran model and Wright-Fisher model deliver similar genealogies, if the parameters are appropriately scaled.
So instead of comparing the theoretical results with results obtained from Moran model simulations, we can check the theoretical results against results
obtained from Wright-Fisher model simulations. 

We consider a 5-kb stretch of DNA in a sample of $n=20$ taken at time of fixation 
of the beneficial allele.
 We divide the stretch into 50 bins, each of length 0.1kb and measure LD between SNPs of two different bins averaged over $10^4$ draws.
Figure \ref{fig:Plot1} shows the results for a recurrent mutation rate $\theta_s=0.1$ and $\theta_s=0.5$, respectively.  
The neutral mutation rate $\theta=0.005$, the recombination rate $\rho= 0.025 $ for two neighboring loci, the distance between neighboring neutral loci $L$ and $R$ is 200bp, 
the selection strength $\alpha=1000$ and the population size $N=10^6$ in both plots. These parameter values are realistic for example for Drosophila melanogaster
samples.

As we see in Figure \ref{fig:Plot1} there is a good fit between simulated and theoretical values of $\widehat{\sigma^2_D}$.
For small sample sizes the extended star-like genealogies approximate the simulated Wright-Fisher-genealogies well. The linkage disequilibrium 
is for theoretical and simulated values high, if $\theta_s \neq0$ and the distance to the selected locus is small, independently of 
the geometry of the considered neutral loci.
Due to recombination linkage disequilibrium decreases with increasing distance to the selected locus. 

The differences between the theoretical and simulated values of $\widehat{\sigma^2_D}$ are due to the approximation of the genealogy. 
The approach has three effects on LD. 

First, a star-like genealogy assumes independent recombinants. But of course in the simulated 
genealogies may also occur coalescence events before recombination events, in particular 
may arise early recombinants (see \citep{DurrettSchweinsberg2004a} or for slightly different
 models \citep{EtheridgePfaffelhuberWakolbinger2006}). 
 For geometry (a) it is important, that recombinants with
 offspring lead to less ``independent'' variation, which can be seen in higher LD values of the simulated data.
For geometry (b) early recombinants become noticeable, because they produce patterns similar to soft sweep patterns: 
Recombinants spread through the populations act as additive founders of the sweep.
For this reason our approximation of the genealogy assumes less founders of the sweep than the simulated genealogies have.
 Therefore the LD-patterns of simulated data should look like the LD-pattern of the theoretical values with a slightly
higher $\theta_s$ value. Higher $\theta_s$ values produce in geometry (b) less
$\widehat{\sigma^2_D}$, compare the pictures in Figure \ref{fig:Plot1}.
This effect becomes negligible for increasing $\theta_s$.    
On the one hand the fixation of the beneficial allele gets faster, on the other hand for intermediate and high values of $\theta$
also extended star-like genealogies assume in average more than two founders. By measuring linkage disequilibrium one can distinguish well
between the existence of one or two founders of the sweep, but not between the existence of three or four founders.
 
Second, the star-like approximation of the genealogy is in general longer 
than the simulated genealogy, since the beneficial allele spreads faster through the population, if the lines are dependent.
 Therefore, more recombination events are assumed to fall on the theoretical genealogies than on the simulated genealogies.
For loci in a small neighborhood of the selected locus this means, that more SNPs can be found for the
star-like genealogy due to recombination. Third, SNPs of simulated data are noisier, they may exist also due to neutral mutation during the sweep.
Both effects can in geometry (a) be recognized by higher theoretical LD values in a small neighborhood of the selected locus. But at a certain point, 
the effect turns over: More recombination brings more ''independent`` variation into the sample: The theoretical values of LD in geometry (a) lie below
 the simulated values.

Often one is interested in the case of a single sample. We simulated a single sample of a soft sweep with two founders and computed linkage 
disequilibrium with that data. The result is plotted in Figure \ref{fig:Singlesample}. In that case the pattern is very clear.
In the neighborhood of the selected locus, high linkage disequilibrium can be found independent of the geometry of the loci.
However, such clear patterns cannot be expected in general, even if the sweep has two founders. It is likely, that the number of offspring is not
 distributed equally between the founders. For example it may happen, that in a sample of 20 individuals with respect
to the beneficial allele 2 individuals are offspring of one founder and the remaining 18 individuals are offspring of the second founder. 
For such unbalanced cases stochastic effects caused by recombination and mutation destroy the pattern easily. 

\section{Discussion}
Soft sweeps have been introduced by Pennings and Hermisson in their series of papers \citep{HermissonPennings2005}, \citep{PenningsHermisson2006a}, \citep{PenningsHermisson2006b}.
They argued, that tests based on haplotype structure have high power to detect soft sweeps.
Linkage disequilibrium is a test sensitive to haplotype structure. If allelic variants are tightly linked to a haplotype, LD
 is high for pairs of such alleles.  
We have seen, that linkage disequilibrium under a non vanishing recurrent mutation rate differs sufficiently
 from linkage disequilibrium under neutrality and hard sweeps, see Figure \ref{fig:Theo}.  

We computed $\widehat{\sigma^2_D}$ to understand
the interplay of haplotype formation due to a soft sweep and recombination.
The main reason to compute $\widehat{\sigma^2_D}$ instead of $\mathbbm{E}[r^2]$ is its mathematical manageability. However, former studies show
 (and the present study
 does that also), that also $\widehat{\sigma^2_D}$ measures what intuitively is understood under linkage disequilibrium and gives a possibility
to distinguish between different population genetics
 scenarios. 
 
When a soft sweep occurs, recombination breaks up the linkage of loci due to haplotype structure. Under hard sweeps recombination causes linkage of
loci lying on one side of the selected locus. In Figure \ref{fig:Theo} theoretical values of $\widehat{\sigma^2_D}$ are plotted
for different values of $\theta_s$ and under neutrality. 
The behavior can be explained roughly in the following manner:

For small values $\theta_s$ we see for both geometries high values of $\widehat{\sigma^2_D}$
 in a small neighborhood of the selected locus decreasing with increasing distance to the selected locus. 
If $\theta_s$ is relatively small, \citet{PenningsHermisson2006a} have shown, that soft sweeps are not very likely, most sweeps will be hard.
LD of hard sweeps depends on 
recombination. Only recombination brings variation into the sample which is necessary to compute linkage disequilibrium. 

After a hard sweep we can see the following pattern of LD due to recombination.
Recombination between the $L$-locus and the $S$-locus includes for geometry (a)
 always a recombination between the $R$-locus and $S$-locus or between the $L$-locus and the $R$-locus, i.e. 
the $L$-locus recombines not independently of the 
$R$-locus. Therefore LD is high for geometry (a) for a hard sweep. 
In geometry (b) an $LS$-recombination event does not cause a $SR$-recombination event and vice versa. So with respect to recombination
 the $L$-locus is
 independent of the $R$-locus. Hence $\widehat{\sigma^2_D}$ is small. If the sample is not finite, $\widehat{\sigma^2_D}$ is
 even zero, see Remark \ref{bemtheo}. 

If a soft sweep occurred, different founders of the sweep bring the variation into the sample - recombination is not necessary.
 If there are exactly
two founders and there exist loci with two allelic variants, such that one allelic variant is carried by one haplotype and the other allele by the other haplotype,
 two of such loci are tightly linked, only recombination can break up this linkage.
 Therefore after a soft sweep with only a few number of founders LD is high in a small neighborhood of the selected locus,
independent of the geometry. But the more founders the soft sweep has, the more variation is in the sample not linked to single founder. This reduces
LD.  

For $r^2$ we expect for very small values of $\theta_s$ patterns of LD similar to hard sweeps, because for very small
values of $\theta_s$ soft sweeps are rare. But $\widehat{\sigma^2_D}$ shows even for very small values of $\theta_s$ high values in a small neighborhood 
of the selected locus. 
This comes from the fact, that small values of $D^2$ expected after a hard sweep in geometry (b) have a smaller effect on the numerator of $\widehat{\sigma^2_D}$ than 
 higher values of $D^2$ expected after a soft sweep in geometry (b). An analogous statement holds for the denominator of $\widehat{\sigma^2_D}$.

For biological studies often the pattern of a single selective sweep is of interest.
After a soft sweep we expect to find high LD of two neutral loci lying in a neighborhood
 of the selected locus, but almost neutral variation. It can be found haplotype structure, where each founder of the sweep gives rise to
 one haplotype.
 In each haplotype group a hard sweep occurred, i.e. almost no variation
can be found, low LD for neutral loci lying on different sides of the selected locus and high LD for loci lying on the same site of the
 selected locus.
In Figure \ref{fig:Singlesample} simulation results of a single sample of a soft sweep with two ancestors with
 respect to the selected locus are shown. As well as \citet{TishkoffEtAl2007} found a comparable clear linkage disequilibrium pattern of a soft sweep in 
 their studies of the human DNA when analyzing the human lactase persistence in African and European human populations.

An adaptation process may not only be initiated by mutation, but also through recurrent migration or from standing genetic 
variation during an environmental change. A two-island model with the beneficial allele fixed in one of the islands and migration
 from this island to the 
other coincides with our model for recurrent mutation.
A more realistic model assumes, that the beneficial allele is not fixed in both islands and that the allelic
frequencies $q_L$, $q_{R}$, etc. do not coincide on both islands. Unfortunately such (simple) modifications make the
 calculations in the proof of Theorem \ref{bigthm}, especially of matrix $A$ and $B$, 
quite complicated.

An improvement of the results could be made by approximating the genealogy not by a star-like approximation but by a marked Yule process
 with immigration. It has been shown by \citet{HermissonPfaffelhuber2008}, that the joint genealogy of the population is better approximated 
by these processes. However, explicit calculations
become with this approximation complicated, since recombination is not independent along lines during the sweep.

\section{Proof of Theorem \ref{bigthm}}\label{sec:proof} 
We proceed in five steps. The quantities $\mathcal{X}_t,\mathcal{Y}_t, \mathcal{Z}_t$ can be expressed
in pairwise heterozygosities. In step 1 we will give this connection. In step 2 we show, how pairwise heterzygosities are transformed to sample heterozygosities.
In step 3 and 4 we show, how pairwise heterozygosities at time $t=T$ are transformed to pairwise heterozygosities at time $t=0$. In step 5 we collect everything together. 

 \medskip

\noindent$\textrm{\textbf{ Step 1}}$:  Link between the pairwise heterozygosities $f_t,g_t,h_t$ and $\mathcal{X}_t$, $\mathcal{Y}_t$, $\mathcal{Z}_t$ \\

\medskip

The quantities $\mathcal{X}_t,\mathcal{Y}_t, \mathcal{Z}_t$ can be expressed in terms of probabilities for pairwise heterozygosities.

 Denote for this purpose by $f_t$ the probability that two pairs heterozygous in both loci are linked, by $g_t$ the probability, that
 exactly one pair of the two pairs
 is linked and
 the other pair is unlinked and by $h_t$ the probability that both pairs are unlinked at time $t$. We can express these probabilities in terms of 
structured partitions: Let $\ell_1, \ell_2 $ be two $L$-loci and $r_1, r_2$ be two $R$-loci taken from the population. Let $\xi_t= (\xi^B_t, \xi^b_t)$ be the genealogy of 
$\{\ell_1, \ell_2, r_1,r_2\}$ at time $t$, then 
$$f_t = P(\xi_t \textrm{ is heterozygous and } \xi^B_t \cup \xi^b_t=\{\{\ell_1,r_1\}, \{\ell_2,r_2\}\}),$$ 
$$g_t = P(\xi_t\textrm{ is heterozygous and } \xi^B_t \cup \xi^b_t=\{\{\ell_1,r_1\}, \{\ell_2\}, \{r_2\}\}),$$
$$h_t = P(\xi_t\textrm{ is heterozygous and } \xi^B_t \cup \xi^b_t=\{\{\ell_1\},\{r_1\}, \{\ell_2\},\{r_2\}\})$$ 

From an easy calculation (see for details also \citep{PfaffelhuberLehnertStephan2008}, Equation (A3)) it follows, that

$$\left(\begin{array}{cc}
	\mathcal{X}_t \\
	\mathcal{Y}_t \\
	\mathcal{Z}_t
	\end{array}\right) =  \frac{1}{4} \underbrace{ \left( \begin{array}{clcr}
               0 & 0    & 1      \\
               0  & 4  &  -4   \\
               1  & -2 & 1
          \end{array}   \right)}_{=:E}   \left(\begin{array}{cc}
	f_t \\
	g_t \\
	h_t
	\end{array}\right)  $$ 

\medskip

\noindent$\textrm{\textbf{ Step 2}}$: Link between pairwise heterozygosities $f,g,h$ and sample
 heterozygosities $\widehat{f}$, $\widehat{g}$, $\widehat{h}$ \\

\medskip
 Denote
 by $\widehat{f}_t$, $\widehat{g}_t$
 and $\widehat{h}_t$ the corresponding sample probabilities, i.e. $\ell_1, \ell_2 \in \dickm \ell$ and $r_1, r_2 \in \dickm r$. It is possible to pick the same individual twice in a sample. Therefore the following relationships hold:    
\begin{align*}
\widehat{f}_t & = \left( 1-\frac{1}{n} \right) f_t \\
\widehat{g}_t & = \left( 1-\frac{1}{n} \right) \left( 1-\frac{2}{n} \right) g_t+ \left( 1-\frac{1}{n} \right) \frac{1}{n} f_t \\
\widehat{h}_t & = \left( 1-\frac{1}{n} \right) \left( 1-\frac{2}{n} \right) \left( 1-\frac{3}{n} \right) h_t+
 \left( 1-\frac{1}{n} \right) \frac{4}{n} \left( 1-\frac{2}{n} \right) g_t + \left( 1-\frac{1}{n} \right) \frac{2}{n}\frac{1}{n}f_t. 
\end{align*}
Denoting 
$$F:=  I + \frac{1}{n}  \left( \begin{array}{clcr}
               -1 & 0    & 0      \\
               1  & -3   & 0   \\
               0  & 4 & -6
          \end{array}   \right)   + \frac{1}{n^2}  \left( \begin{array}{clcr}
               0 & 0    & 0      \\
               -1  & 2   & 0   \\
               2  & -12 & 11
          \end{array}   \right)  + \frac{1}{n^3}  \left( \begin{array}{clcr}
               0 & 0    & 0      \\
               0  & 0   & 0   \\
               -2  & 8 & -6
          \end{array}   \right) ,$$

this is equivalent to 
$\left(\begin{array}{cc}
	\widehat{f_t} \\
	\widehat{g_t} \\
	\widehat{h_t}
	\end{array}\right) = F  \left(\begin{array}{cc}
	f_t \\
	g_t \\
	h_t
	\end{array}\right)$

For example, two linked pairs of one allele at the $L$- and one allele at the
$R$-locus each taken at random (with replacement) from a sample are
heterozygous, if we did not pick the same individual twice and the
resulting two different lines are heterozygous at both loci.

\medskip

In the next two steps we compute how to find $f_0,g_0$ and $h_0$ given $f_T,g_T$ and $h_T$, respectively.\\

\medskip 

\noindent$\textrm{\textbf{ Step 3}}$ : From $f_{T/2},g_{T/2}$ and $h_{T/2}$ to $f_0,g_0$ and $h_0$, respectively \\

\medskip

For both geometries we have
 $$\left(\begin{array}{cc}
	f_0 \\
	g_0 \\
	h_0
	\end{array}\right) = C  \left(\begin{array}{cc}
	f_{T/2} \\
	g_{T/2} \\
	h_{T/2}
	\end{array}\right)$$ with
$$C =  \left( \begin{array}{clcr}
               {p_{LR}}^2 & 2p_{LR}(1-p_{LR})    & (1-p_{LR})^2      \\
               0  & p_{LR}   & 1-p_{LR}   \\
               0  & 0 & 1
          \end{array}   \right)$$

Our model assumptions coincide with the model assumptions of \citet{PfaffelhuberLehnertStephan2008} in the time interval $[T,T/2)$, so that we obtain the same results here. \\

\medskip

\noindent$\textrm{\textbf{ Step 4}}$ : From $f_{T},g_{T}$ and $h_{T}$ to  $f_{T/2},g_{T/2}$ and $h_{T/2}$, respectively\\

\medskip

For this time step it is important to note, that it has to be paid attention not only on the two neutral loci, but also
on the selected locus. We use Ewens sampling formula (see Equation \ref{probfounder}) to compute the probabilities, if the ancestral lines of the pairs
 share with respect to the selected locus a 
common ancestors or different ancestors. With this we get the following relationships:

 $$\left(\begin{array}{cc}
	f_{T/2} \\
	g_{T/2} \\
	h_{T/2}
	\end{array}\right) = A  \left(\begin{array}{cc}
	f_{T} \\
	g_{T} \\
	h_{T}
	\end{array}\right) \qquad \textrm{ for geometry (a) } $$
and 

 $$\left(\begin{array}{cc}
	f_{T/2} \\
	g_{T/2} \\
	h_{T/2}
	\end{array}\right) = B  \left(\begin{array}{cc}
	f_{T} \\
	g_{T} \\
	h_{T}
	\end{array}\right) \qquad
\textrm{ for geometry (b) }$$

with matrix $$A= (a_{ij})_{1\leq i,j\leq 3}$$  given by

\begin{align*}a_{11} & = \frac{\theta_s}{\theta_s +1} p_{LR}^2 + \frac{1}{\theta_s +1}(1-p_{SL}^2)p_{LR}^2  \\
a_{12} & = \frac{\theta_s}{\theta_s +1} 2p_{LR}(1-p_{LR})+ \frac{1}{\theta_s +1} 2 (1-p_{SL}^2)p_{LR} (1-p_{LR}) \\
a_{13} & = \frac{\theta_s}{\theta_s +1} (1-p_{LR})^2  + \frac{1}{\theta_s +1} (1-p_{SL}^2)(1-p_{LR})^2\\
a_{21} & = \frac{\theta_s}{(\theta_s+1)(\theta_s+2)} p_{SL}p_{SR}p_{LR} + \frac{2}{(\theta_s+1)(\theta_s+2)} (1-p_{SL}) p_{SL}^2 p_{LR}^2 \\
a_{22} & = \frac{{\theta_s}^2}{(\theta_s +1)(\theta_s +2)} p_{LR} + \frac{\theta_s}{(\theta_s +1)(\theta_s+2)}(3p_{LR} - 3 p_{SR} p_{SR}) + \\
&\quad + \frac{2}{(\theta_s+2)(\theta_s+1)} p_{LR}(1-p_{SL})(1+p_{SL}-4p_{SL}p_{SR} + 2p_{SL}^2) \\
a_{23} & = \frac{\theta_s^2}{(\theta_s +1)(\theta_s +2)}(1-p_{LR}) + \\
& \quad + \frac{\theta_s}{(\theta_s+1)(\theta_s+2)} (1-p_{LR})(1-p_{SL}^2 + 1- p_{SL}p_{SR}+ 1- p_{SL}p_{SR}) +\\
& \quad + \frac{2}{(\theta_s+1)(\theta_s+2)} (1-p_{SL})(1-p_{LR})(1+p_{SL}- 2p_{SL}p_{SR}) \\
a_{31} & =  \frac{2\theta_s}{(\theta_s+1)(\theta_s+2)(\theta_s+3)} p_{SL}^2 p_{SR}^2  \\
a_{32} & = \frac{ 4 \theta_s^2}{(\theta_s+1)(\theta_s+2)(\theta_s+3)} p_{SL}p_{SR}  + \\
& \quad + \frac{2\theta_s}{(\theta_s+1)(\theta_s+2)(\theta_s+3)}(p_{SR}p_{SL}(1-p_{SR}p_{SL})) + \\
& \quad + \frac{4 \theta_s}{(\theta_s+1)(\theta_s+2)(\theta_s+3)} p_{SL}p_{SR}(2-p_{SL}-p_{SR})+ \\
& \quad + \frac{3!}{(\theta_s+1)(\theta_s+2)(\theta_s+3)} (4 (1-p_{SL})(p_{SL}(1-p_{SR})p_{SR}))\\
a_{33} & =  \frac{ \theta_s^3 }{(\theta_s + 1)(\theta_s + 2)(\theta_s +3)} + \\ 
& \quad + \frac{\theta_s^2}{(\theta_s + 1)(\theta_s + 2)(\theta_s + 3)} (4(1 - p_{SL}p_{SR}) + (1 - p_{SL}p_{SL}) + (1 - p_{SR}p_{SR})) + \\
          & \quad + \frac{4\theta_s}{(\theta_s + 1)(\theta_s + 2)(\theta_s + 3)}((1 -p_{SR})(1-p_{SL})(1-p_{SR})+ (1 - p_{SR})(1 -p_{SR})p_{SL}) + \\
& \quad + \frac{4\theta_s}{(\theta_s + 1)(\theta_s + 2)(\theta_s + 3)}( 2(1-p_{SR})(1-p_{SL})p_{SR} + (1-p_{SL})(1-p_{SL})p_{SR})+\\
&\quad  + \frac{4\theta_s}{(\theta_s + 1)(\theta_s + 2)(\theta_s + 3)} ( 2(1-p_{SR})p_{SL}(1-p_{SL}) + (1-p_{SL})(1-p_{SL})(1-p_{SR})) + \\
& \quad + \frac{\theta_s}{(\theta_s + 1)(\theta_s + 2)(\theta_s + 3)}(2(1 - p_{SL}p_{SR})(1 - p_{SL}p_{SR}) +  (1 - p_{SL}p_{SL})(1 - p_{SR}p_{SR}) ) + \\
& \quad +  \frac{6}{(\theta_s + 1)(\theta_s + 2)(\theta_s + 3)}((1 - p_{SL})(1 - p_{SR})(1 +  p_{SL} + p_{SR} - 3p_{SL}p_{SR}))
\end{align*}

and matrix $$B= (b_{ij})_{1\leq i,j \leq 3}$$ given by 

\begin{align*}
b_{11} & =  \frac{\theta_s}{\theta_s+1} p_{LS}^2 p_{SR}^2 \\
b_{12} & =\frac{2\theta_s}{\theta_s+1} p_{SR} p_{LS}(1-p_{LS}p_{SR}) + \frac{4}{\theta_s+1} p_{LS}(1-p_{LS})p_{SR} (1-p_{SR}) \\ 
b_{13} & =  \frac{\theta_s}{\theta_s+1} (1-p_{LS}p_{SR})^2 + \\
& \quad + \frac{1}{\theta_s+1} (1-p_{LS})(1-p_{SR})(2p_{LS}(1-p_{SR})+2(1-p_{LS})p_{SR} +(1-p_{LS})(1-p_{SR})) \\
b_{21} & = \frac{\theta_s}{(\theta_s+1)(\theta_s+2)}p_{LS}^2p_{SR}^2 \\
b_{22} & =  \frac{\theta_s^2}{(\theta_s+1)(\theta_s+2)} p_{LS} p_{SR} + \\
& \quad + \frac{\theta_s}{(\theta_s+1)(\theta_s+2)}(p_{LS}p_{SR} (1-p_{LS}p_{SR})  +  (2-p_{SR} -p_{LS})p_{SR}p_{LS} + p_{SR} p_{LS}(1-p_{LS}p_{SR})) + \\
& \quad + \frac{8}{(\theta_s+1)(\theta_s+2)} p_{LS}(1-p_{LS})p_{SR} (1-p_{SR})\\
b_{23} & =  \frac{\theta_s^2}{(\theta_s+1)(\theta_s+2)}(1-p_{SR}p_{LS}) + \\
& \quad + \frac{\theta_s}{(\theta_s+1)(\theta_s+2)}((1-p_{LS})^2(1-p_{SR}) + 2p_{LS}(1-p_{LS})(1-p_{SR})) + \\
& \quad + \frac{\theta_s}{(\theta_s+1)(\theta_s+2)}( p_{SR}(1-p_{LS})^2 + (1-p_{LS}p_{SR})^2 + (1-p_{SR})^2(1-p_{LS})) + \\
& \quad + \frac{\theta_s}{(\theta_s+1)(\theta_s+2)}(  2 p_{SR}(1-p_{SR})(1-p_{LS})+ p_{LS}(1-p_{SR})^2) \\
& \quad + \frac{2}{(\theta_s+1)(\theta_s+2)} (1-p_{LS})(1-p_{SR})(2p_{LS}(1-p_{SR}))\\
& \quad + \frac{2}{(\theta_s+1)(\theta_s+2)}(2(1-p_{LS})p_{SR} +(1-p_{LS})(1-p_{SR}))\\
b_{31} & = a_{31} \qquad \textrm{ with } p_{LS} \textrm{ instead of } p_{SL}  \\
b_{32} & = a_{32} \qquad \textrm{ with } p_{LS} \textrm{ instead of } p_{SL}  \\
b_{33} & = a_{33} \qquad \textrm{ with } p_{LS} \textrm{ instead of } p_{SL} 
\end{align*}

To see the above equations, consider for example in geometry (a) the term $a_{21}$:
 In this case at time $T/2$ there are two pairs, which are heterozygous in both
loci and exactly one of the pairs is linked. If the two pairs have two different ancestors
 with respect to the selected locus neither a $SL$-recombination nor a $SR$-recombination must happen for the unlinked pair, nor
 a $LR$-recombination event to the linked pair. Therefore the probability to stay linked also at the beginning of the sweep is 
 $\frac{\theta_s}{(\theta_s+1)(\theta_s+2)} p_{SL}p_{SR}p_{LR}$ using Ewens sampling
 formula.
 If the two pairs have a single ancestor, the linked pair has to change backgrounds, 
i.e. a $SL$-recombination event has to take place. Therefore we obtain in this case the probability $\frac{2}{(\theta_s+1)(\theta_s+2)} (1-p_{SL}) p_{SL}^2 p_{LR}^2$.
 The sum of these two probabilities gives $a_{21}$. 
The other terms can be explained in an analogous manner.\\ 

\medskip

\noindent$\textrm{\textbf{ Step 5}}$ : Collecting all together

\medskip 

We have

$$(\widehat{\mathcal{X}_0}, \widehat{\mathcal{Y}_0}, \widehat{\mathcal{Z}_0})^{T} =E \cdot  F \cdot C\cdot A \cdot E^{-1} (\mathcal{X}_T, \mathcal{Y}_T, \mathcal{Z}_T)^{T}$$ 
for geometry (a) and 

$$(\widehat{\mathcal{X}_0}, \widehat{\mathcal{Y}_0},\widehat{ \mathcal{Z}_0})^{T} = E \cdot F \cdot C\cdot B \cdot E^{-1}
 (\mathcal{X}_T, \mathcal{Y}_T, \mathcal{Z}_T)^{T} $$
for geometry (b). 

With this we can compute $\widehat{\sigma^2_D} = \widehat{\mathcal{Z}_0}/ \widehat{\mathcal{X}_0}$ (recalling
 Equation (\ref{sigma})).

A calculation with Mathematica gives Equations (\ref{big1})-(\ref{big4}), if terms of order $\theta_s^2$ and $1/n$ are ignored. 
 $\square$

\subsubsection*{Acknowledgement}
I am grateful to Peter Pfaffelhuber for many fruitful discussions.
 Many thanks to Greg Ewing for providing and helping me with msms and to Franz Baumdicker and Joachim Hermisson for helpful comments on the manuscript.
I acknowledge support from the DFG Forschergruppe 1078 "Natural
selection in structured populations".

\bibliography{HP}
\bibliographystyle{apalike}

\end{document}